\documentclass[twocolumn, 10pt]{IEEEtran}
\usepackage{setspace}
\usepackage{color}
\usepackage{cite}
\usepackage{graphicx}
\usepackage{subfigure}
\usepackage{amsmath}
\usepackage{amssymb}
\usepackage{amsthm}
\usepackage{amsmath}
\usepackage{algorithm,algorithmic}
\usepackage{booktabs}
\usepackage{multirow}

\newcommand{\bdsb}[1]{\boldsymbol{#1}}

\newtheorem{lem}{Lemma}

\newtheorem{corollary}{Corollary}
\newtheorem{thm}{Theorem}

\newtheorem{condition}{Condition}

\title{Robust Decentralized State Estimation and Tracking for Power Systems via Network Gossiping}

\author{Xiao Li, {\it Student Member, IEEE} and Anna Scaglione, {\it Fellow, IEEE}
\thanks{This work was supported by the U.S. Department of Energy through the Trustworthy Cyber Infrastructure for Power Grid (TCIPG) program.}
\thanks{%
  The authors are with the Dept. of Electrical and Computer Engineering,
  University of California, Davis.
  (email : \{eceli,ascaglione\}@ucdavis.edu).}
\thanks{
Part of this work was presented in \cite{li2012decentralized} at IEEE SAM 2012 and in \cite{li2013robust} at IEEE ICASSP 2013.
}
}

\graphicspath{{./figure/}}

\begin{document}

\maketitle

\begin{abstract}
This paper proposes a fully decentralized adaptive re-weighted state estimation (DARSE) scheme for power systems via network gossiping. The enabling technique is the proposed Gossip-based Gauss-Newton (GGN) algorithm, which allows to harness the computation capability of each area (i.e. a database server that accrues data from local sensors) to collaboratively solve for an accurate global state. The DARSE scheme mitigates the influence of bad data by updating their error variances online and re-weighting their contributions adaptively for state estimation. Thus, the global state can be estimated and tracked robustly using near-neighbor communications in each area. Compared to other distributed state estimation techniques, our communication model is flexible with respect to reconfigurations and resilient to random failures as long as the communication network is connected. Furthermore, we prove that the Jacobian of the power flow equations satisfies the Lipschitz condition that is essential for the GGN algorithm to converge to the desired solution. Simulations of the IEEE-118 system show that the DARSE scheme can estimate and track online the global power system state accurately, and degrades gracefully when there are random failures and bad data.
\end{abstract}



\section{Introduction}
Traditionally, power system state estimation (PSSE) has been solved by the iterative Gauss-Newton (GN) procedure (or its variants) using measurements from Supervisory Control and Data Acquisition (SCADA) systems \cite{schweppe1974static}. Recently, Phasor Measurement Units (PMU) in the Wide-Area Measurement Systems (WAMS) are gaining increasing attention since state estimation using PMU data becomes a linear least squares problem \cite{yang2011transition}. However, due to the limited deployment of PMUs, researchers have proposed hybrid estimation schemes \cite{phadke1986state,zivanovic1996implementation} by integrating WAMS and SCADA measurements. Some of these methods incorporate the PMU measurements into the GN update \cite{meliopoulos2010supercalibrator,qin2007hybrid,nuqui2007hybrid}, while others use PMU data to further refine the estimates from SCADA data \cite{chakrabarti2010comparative,avila2009recent}.

\vspace{-0.3cm}
\subsection{Related Work}
Driven by the ongoing expansion of the cyber infrastructure of power systems and fast sampling rates of PMU devices, there has been growing concern on distributing the computations across different areas by fusing information in various ways \cite{brice1982multiprocessor,kurzyn1983real,yang2011transition,gomez2011multilevel,falcao1995parallel,lin1992distributed,ebrahimian2000state,van1981two,zhao2005multi,jiang2007distributed}. In most of these methods, each distributed area solves for a local state and refines the local estimates in a hierarchical manner by leveraging on the tie-line structure and tuning the estimates on boundary buses shared with neighboring areas. Although these methods considerably alleviate the computational burdens at control centers, they rely on {\it aggregation trees} that require centralized coordination and depend on the power grid topology. This is clearly a limitation in reconfiguring the system, if random failures or attacks call for it. Last but not least, this class of methods typically requires local observability provided by redundant local measurements, which may not be satisfied during contingencies.

Recently, the authors of \cite{xie2012fully,kekatos2012distributed} proposed distributed state estimation schemes for power systems that do not require local observability. Specifically,  \cite{kekatos2012distributed} follows a similar formulation as in \cite{falcao1995parallel,lin1992distributed,ebrahimian2000state} and uses the {\it alternating direction method of multipliers} (AD-MoM) to distribute the state estimation procedure. The merit of AD-MoM in \cite{kekatos2012distributed} lies in the fact that the computation is local and scalable with respect to the number of variables in each area. However, the communications required by the AD-MoM scheme are constrained by the grid topology. 
Also, the numerical tests performed in \cite{kekatos2012distributed} are based exclusively on PMU data and the algorithm convergence in the presence of SCADA measurements is not discussed in general.
While the advantages of hybrid state estimation schemes are evident from \cite{phadke1986state,zivanovic1996implementation}, these papers  do not provide analytical proof nor performance guarantee for the convergence.

In order to obtain the global state, the approach taken in \cite{xie2012fully} is inspired by a number of recent advances in network diffusion algorithms for optimization. Diffusion algorithms are capable of solving an array of problems  in a fully decentralized manner without any hierarchical aggregation, including linear filtering \cite{lopes2008diffusion}, convex optimizations \cite{nedic2009distributed} and adaptive estimation \cite{kar2008distributed}. These techniques combine a local descent step with a diffusion step, which is performed via {\it network gossiping}. The convergence of these algorithms relies on the convexity of the cost function and a small (or diminishing) step-size, which slows down the algorithm in general. Furthermore, PSSE with SCADA measurements is non-convex and it is not clear how these methods will perform in practice. 

Compared to other decentralized methods including \cite{li2012decentralized,li2013robust}, another major issue addressed in this paper is bad data processing. There has been extensive work devoted to the detection and identification of bad data in power systems, mainly divided into two categories. The first category is usually handled by $\chi^2$-test and the largest normalized residual (LNR) test \cite{van1985bad,garcia1979fast,bobba2010detecting,liu2011false,schweppe1974static}. There are also specific detection schemes using PMU measurements in different ways \cite{chen2006placement,giani2011smart}. To reduce the bad data effects, LNR tests are performed successively, which re-estimates the state after removing or compensating the data entry with the largest residual \cite{ebrahimian2000state,monticelli1999state}. Based on the LNR principle, \cite{choi2011fully} developed a distributed scheme where different areas exchange measurement residuals and successively re-estimate the state, until no further alarms are triggered. To avoid repetitive computations, \cite{bin1994implementable,pasqualetti2011distributed} suggested computing $\chi^2$ and LNR test statistics as a rule-of-thumb to identify bad data in one pass, where the state only needs to be re-estimated one additional time. Furthermore, the work \cite{kekatos2012distributed,xu2011sparse} proposed a convex optimization approach to directly estimate the bad data jointly with the state variables by integrating a sparsity constraint on the number of error outliers. The second approach for bad data processing, on the other hand, is to suppress their effects on the state estimates instead of removal.  For instance, \cite{merrill1971bad, falcao1982power} propose incorporating different weights for the residuals to limit the impacts of bad data on the state updates. A more comprehensive literature review on bad data processing can be found in \cite{lo1983development,handschin1975bad}.

\vspace{-0.3cm}
\subsection{Contributions}
In this paper, we formulate the state estimation problem in a Maximum Likelihood (ML) framework, and develop the {\it Decentralized Adaptive Re-weighted State Estimation} (DARSE) scheme specifically for wide-area PSSE. The DARSE scheme deals with bad data similarly to \cite{merrill1971bad, falcao1982power}, where the bad data variances are adaptively updated based on the measurements. Furthermore, the DARSE scheme generalizes the Gossip-based Gauss-Newton (GGN) algorithm we proposed in \cite{li2012convergence} under an adaptive setting, which exhibits faster convergence than the distributed state estimation scheme in \cite{xie2012fully} derived from first order diffusion algorithms. Another important contribution of this paper is that we prove sufficient conditions for the convergence of GGN algorithm, by showing that the Jacobian of power flow equations satisfies strictly the Lipschitz condition. Furthermore, thanks to the adaptive features of the DARSE scheme, it automatically adjusts the weights for different sensor observations based the measurement quality, to reduce the impacts of bad data on the state estimates. The main benefit of the DARSE scheme is that it is completely adaptable to both time-varying measurements quality and network conditions. The quality of measurements can degrade in a random fashion and also the communication network can experience failures. With mild connectivity conditions, DARSE is able to deliver accurate estimates of the global state at each distributed area. Our claims are verified numerically on the IEEE-118 system.

\vspace{-0.3cm}
\subsection{Notations}
In this paper, we used the following notations:
\begin{itemize}
	\item $\mathrm{i}$: imaginary unit and $\mathbb{R}$ and $\mathbb{C}$: real and complex numbers.
	\item $\Re\{\cdot\}$ and $\Im\{\cdot\}$: the real and imaginary part of a number.
	\item $\mathbf{I}_N$: an $N\times N$ identity matrix.
	\item $\mathbf{1}_N$: an $N\times 1$ vector with all entries equal $1$.
	\item $\|\mathbf{A}\|$ and $\|\mathbf{A}\|_F$ are the $2$-norm\footnote{The $2$-norm of a matrix is the maximum of the absolute value of the eigenvalues and the $2$-norm of a vector $\mathbf{x}\in\mathbb{R}^N$ is $\|\mathbf{x}\|=\sqrt{\sum_{n=1}^N x_n^2}$.} and $F$-norm of a matrix.
	\item $\mathrm{vec}(\mathbf{A})$ is the vectorization of a matrix $\mathbf{A}$.
	\item $\mathbf{A}^T$, $\mathrm{Tr}(\mathbf{A})$, $\lambda_{\min}(\mathbf{A})$ and $\lambda_{\max}(\mathbf{A})$: transpose, trace, minimum and maximum eigenvalues of matrix $\mathbf{A}$.
	\item $\otimes$ is the Kronecker product and $\mathbb{E}[\cdot]$ means expectation.
\end{itemize}

\section{Power System State Estimation}
The power grid is characterized by {\it buses} that represents interconnections, generators or loads, denoted by the set $\mathcal{N}\triangleq\{1,\cdots,N\}$. The grid topology is determined by the edge set $\mathcal{E} \triangleq \{\{n,m\}\}$ with cardinality $|\mathcal{E}|=E$, which corresponds to the transmission line between bus $n$ and $m$. The Energy Management Systems (EMS) collect measurements on certain buses and transmission lines to estimate the state of the power system, i.e., the voltage phasor $V_n\in\mathbb{C}$ at each bus $n\in\mathcal{N}$. In this paper, we consider the Cartesian coordinate representation using the real and imaginary components of the complex voltage phasors $\mathbf{v}=[\Re\{V_1\},\cdots,\Re\{V_N\},\Im\{V_1\},\cdots,\Im\{V_N\}]^T$. This representation facilitates our derivations because it expresses PMU measurements as a linear mapping and SCADA measurements as quadratic forms of the state $\mathbf{v}$ as in \cite{lavaei2010zero}.

\vspace{-0.3cm}
\subsection{Measurement Model}
Since there are 2 complex injection measurements at each bus and 4 complex flow measurements on each line, this amounts to twice as many real variables. Thus the measurement ensemble has $M=4N+8E$ entries in an aggregate vector partitioned into four sections\footnote{Subscripts $\{\mathcal{V},\mathcal{C},\mathcal{I},\mathcal{F}\}$ mean voltage, current, injection and flow.}
\begin{align}
	 \mathbf{z}[t]=[\mathbf{z}_{\mathcal{V}}^T[t],\mathbf{z}_{\mathcal{C}}^T[t],\mathbf{z}_{\mathcal{I}}^T[t],\mathbf{z}_{\mathcal{F}}^T[t]]^T,
\end{align}	
containing the length-$2N$ voltage phasor $\mathbf{z}_{\mathcal{V}}[t]$ and power injection vector $\mathbf{z}_{\mathcal{I}}[t]$ at bus $n\in\mathcal{N}$, the length-$4E$ current phasor $\mathbf{z}_{\mathcal{C}}[t]$ and power flow vector $\mathbf{z}_{\mathcal{F}}[t]$ on line $(n,m)\in\mathcal{E}$ at bus $n$. These measurements are gathered at a certain periodicity $T_{\textrm{\tiny SE}}$ for state estimation.
In contrast to the slow rate SCADA measurements, since PMU devices also provide fast samples for dynamic monitoring and control, some pre-processing is needed to align measurements that come at widely different rates. This is an important practical issue prior to the state estimation, however it is unrelated with the estimation methodology considered here and therefore is left for future investigation. State estimation is performed using SCADA measurements $\left\{\mathbf{z}_{\mathcal{I}}[t],\mathbf{z}_{\mathcal{F}}[t]\right\}$ and PMU measurements 
$\left\{\mathbf{z}_{\mathcal{V}}[t],\mathbf{z}_{\mathcal{C}}[t]\right\}$ that have been pre-processed and aligned.

Defining the power flow equations $\mathbf{f}_{(\cdot)}(\mathbf{v})$ in Appendix \ref{power_flow_equations_appendix} and letting $\bar{\mathbf{v}}[t]$ be the true state at time $t$, the individual vector $\mathbf{z}_{(\cdot)}[t]=\mathbf{f}_{(\cdot)}(\mathbf{\bar{v}}[t])+ \mathbf{r}_{(\cdot)}[t]$ contains observations corrupted by measurement noise $\mathbf{r}_{(\cdot)}[t]$ that arises from instrumentation imprecision and random outliers whose variances are potentially much larger due to attacks or equipment malfunction. The entries that have large variances are what we call {\it bad data}. Then we have the measurement model below
\begin{align}\label{meas-model_all}
	\mathbf{z}[t] = \mathbf{f}(\mathbf{\bar{v}}[t]) + \mathbf{r}[t],
\end{align}
where
\begin{align}
	\mathbf{r}[t] &= [\mathbf{r}_{\mathcal{V}}^T[t],\mathbf{r}_{\mathcal{C}}^T[t],\mathbf{r}_{\mathcal{I}}^T[t],\mathbf{r}_{\mathcal{F}}^T[t]]^T\\
	\mathbf{f}(\mathbf{v}) &= [\mathbf{f}_{\mathcal{V}}^T(\mathbf{v}),\mathbf{f}_{\mathcal{C}}^T(\mathbf{v}),\mathbf{f}_{\mathcal{I}}^T(\mathbf{v}),\mathbf{f}_{\mathcal{F}}^T(\mathbf{v})]^T.
\end{align}	

A practical data collection architecture in power systems (compatible with WAMS and SCADA) consists of $I$ interconnected {\it areas}, where each {\it area} records a subset of $\mathbf{z}$ in \eqref{meas-model_all}. To describe the measurement selection in each area, we define a binary mask $\mathbf{T}_{i,(\cdot)}\in\{0,1\}^{M_{i,(\cdot)}\times M}$ with $M_i$ rows and $M$ columns, where each row is a canonical basis vector $\mathbf{e}_m=[0,\cdots,1,\cdots,0]^T$ picking the corresponding element in each category from $\{\mathcal{V},\mathcal{C},\mathcal{I},\mathcal{F}\}$. Letting $$\mathbf{T}_i\triangleq [\mathbf{T}_{i,\mathcal{V}}^T, \mathbf{T}_{i,\mathcal{C}}^T, \mathbf{T}_{i,\mathcal{I}}^T,\mathbf{T}_{i,\mathcal{F}}^T]^T$$ be the selection matrix in the $i$-th area and applying this mask on $\mathbf{z}$, the measurements in the $i$-th area are selected as
\begin{align}\label{meas-model}
	\mathbf{c}_i[t] = \mathbf{f}_i(\bar{\mathbf{v}}[t]) + \mathbf{r}_i[t],
\end{align}
where $\mathbf{c}_i[t] \triangleq \mathbf{T}_i\mathbf{z}[t] = [\mathbf{c}_{i,\mathcal{V}}^T[t], \mathbf{c}_{i,\mathcal{C}}^T[t],\mathbf{c}_{i,\mathcal{I}}^T[t],\mathbf{c}_{i,\mathcal{F}}^T[t]]^T$ and similarly $\mathbf{f}_i(\cdot)=\mathbf{T}_i\mathbf{f}(\cdot)$, $\mathbf{r}_i[t]=\mathbf{T}_i\mathbf{r}[t]$. We assume that $\mathbf{r}_i[t]$'s are Gaussian and uncorrelated between different areas, which has an unknown covariance denoted by
\begin{align}
	\mathbf{R}_i[t]\triangleq \mathrm{diag}[\epsilon_{i,1}[t],\cdots,\epsilon_{i,M_i}[t]],
\end{align}	
where $M_i=M_{i,\mathcal{V}}+M_{i,\mathcal{C}}+M_{i,\mathcal{I}}+M_{i,\mathcal{F}}$ is the total number of observations from each type of measurements.


\begin{figure*}
\centering
\includegraphics[width=0.8\linewidth]{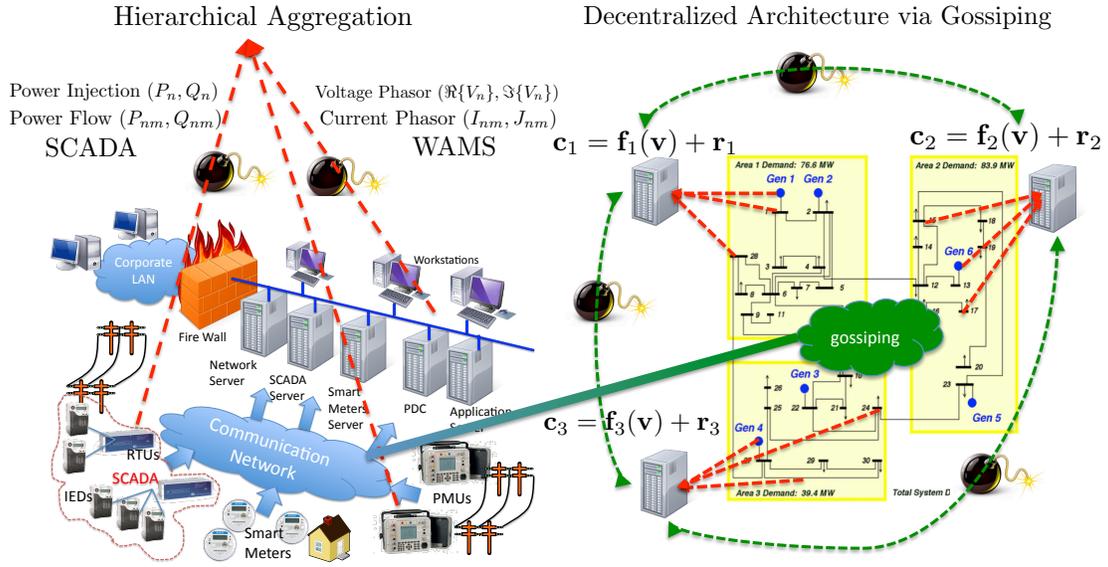}
\caption{Centralized architecture v.s. the decentralized architecture for DARSE.}\label{fig.architecture}
\vspace{-0.2cm}
\end{figure*}

\vspace{-0.3cm}
\subsection{Maximum Likelihood (ML) Estimation}\label{ML_estimation}
Using the measurement model in \eqref{meas-model}, the ML estimate of the state is obtained by maximizing the likelihood function with unknown noise covariances $\left\{\mathbf{R}_i[t]\right\}_{i=1}^I$ over the state space $\mathbb{V}$
\begin{align*}
	\left\{\widehat{\mathbf{v}}[t],\widehat{\mathbf{R}}_i[t]\right\}
	&=
	\arg\underset{\mathbf{v}\in\mathbb{V},\mathbf{R}_i}{\max}~\log \mathbb{P}\left(\left\{\mathbf{c}_i[t]\right\}_{i=1}^I\Big|\,\mathbf{v},\left\{\mathbf{R}_i[t]\right\}_{i=1}^I\right),
\end{align*}
which is equivalent to minimizing the following cost function
\begin{align}
	\underset{\mathbf{v}\in\mathbb{V},\{\mathbf{R}_i\}_{i=1}^I}{\min}~
	 \sum_{i=1}^I \left\|\mathbf{c}_i[t]-\mathbf{f}_i(\mathbf{v})\right\|_{\mathbf{R}_i^{-1}}^2
	 +\sum_{i=1}^I\log|\mathbf{R}_i|. \label{ML_joint}
\end{align}
Note that due to the quadratic nature of power flow equations in Appendix \ref{power_flow_equations_appendix}, the cost function \eqref{ML_joint} is highly non-convex in $\mathbf{v}$ and there exist multiple stationary points $\mathbf{v}^\star$ in the set $\mathbb{V}^\star$ that result in local minima of \eqref{ML_joint}. According to \cite{kayfundamentals}, these stationary points can be determined by setting the derivatives with respect to $\mathbf{v}$ and $\{\epsilon_{i,m}\}_{m=1,\cdots,M_i}^{i=1,\cdots,I}$ to zero
\begin{align}
	 \sum_{i=1}^I\mathbf{F}_i^T(\mathbf{v}^\star)\mathbf{R}_i^\star\left[\mathbf{c}_i[t]-\mathbf{f}_i(\mathbf{v}^\star)\right] &= \mathbf{0},\quad \mathbf{v}^\star\in\mathbb{V}^\star \label{x_opt_1}\\
	 |c_{i,m}[t]-f_{i,m}(\mathbf{v}^\star)|^2 &= \epsilon_{i,m}^\star. \label{covariance_est}
\end{align}
Clearly, the ML estimate $\widehat{\mathbf{v}}[t]$ achieve the global minimum of \eqref{ML_joint} and belongs to the set $\mathbb{V}^\star$. Solving for the ML estimate $\widehat{\mathbf{v}}[t]$ would require substituting $\epsilon_{i,m}^\star$ (with the unknown $\mathbf{v}^\star$) back into $\mathbf{R}_i^\star$ in \eqref{x_opt_1} and searching for the point $\mathbf{v}^\star$ that achieves the global minimum of \eqref{ML_joint}. Therefore, the joint ML estimates are obtained by equivalently solving the coupled equations
\begin{align}
	 \widehat{\mathbf{v}}[t] &= \arg\underset{\mathbf{v}\in\mathbb{V}}{\min}
	 \sum_{i=1}^I \left\|\mathbf{c}_i[t]-\mathbf{f}_i(\mathbf{v}[t])\right\|_{\widehat{\mathbf{R}}_i^{-1}[t]}^2
	 \label{x_opt}\\
	 \widehat{\epsilon}_{i,m}[t] &=  |c_{i,m}[t] - f_{i,m}(\widehat{\mathbf{v}}[t])|^2 \label{variance_opt}\\
	 \widehat{\mathbf{R}}_i[t] &=\mathrm{diag}[\cdots,\widehat{\epsilon}_{i,m}[t],\cdots]
\end{align}
where we have used the fact that given the ML estimate of the noise variance $\widehat{\epsilon}_{i,m}[t]$, the stationary points of the non-linear least squares (NLLS) problem in \eqref{x_opt} are identical to that in \eqref{x_opt_1}. This equivalence can be regarded as the non-linear version of the linear estimation with nuisance parameters in \cite{kayfundamentals}. However, this approach is still highly non-linear and complex. In the following, we take advantage of the streaming measurements to switch adaptively between estimating the state \eqref{x_opt} and estimating the variances \eqref{variance_opt}.

\vspace{-0.3cm}
\subsection{Adaptive Re-weighted State Estimation}\label{ARSE_sec}
If the noise covariance is known, the state \eqref{x_opt} can be obtained directly from conventional PSSE \cite{schweppe1974static}. Therefore, we propose to use the previous covariance estimate as a substitute\footnote{In general, a better substitute can be predicted using the temporal statistics of the random process $\mathbf{r}[t]$, but here we simply use the previous estimate.} of $\widehat{\mathbf{R}}_i[t]$, $i=1,\cdots,I$ to re-weight the measurements in the current snapshot, and propose the {\it Adaptive Re-weighted State Estimation} (ARSE) scheme\footnote{If desired, one can iterate once again the state estimation after the outlier covariance has been updated to give a better state.} in Algorithm \ref{ARGN_algorithm}.

\begin{algorithm}[t]
\caption{ARSE Scheme}\label{ARGN_algorithm}
\begin{algorithmic}[1]
\STATE Predict outlier covariance $\bdsb{\Gamma}_i = \widehat{\mathbf{R}}_i[t-1]$, $i=1,\cdots,I$
\STATE Update state estimates
\begin{align*}
	\widehat{\mathbf{v}}[t] &= \arg\underset{\mathbf{v}\in\mathbb{V}}{\min}  \sum_{i=1}^I \left[\mathbf{c}_i[t]-\mathbf{f}_i(\mathbf{v})\right]^T\bdsb{\Gamma}_i^{-1}\left[\mathbf{c}_i[t]-\mathbf{f}_i(\mathbf{v})\right]
\end{align*}
\STATE Adjust covariance $\widehat{\mathbf{R}}_i[t] = \mathrm{diag}[\widehat{\epsilon}_{i,1}[t],\cdots,\widehat{\epsilon}_{i,M_i}[t]]$
\begin{align}\label{update_gross_error_stat}
	\widehat{\epsilon}_{i,m}[t] =  |c_{i,m}[t] - f_{i,m}(\widehat{\mathbf{v}}[t])|^2.
\end{align}
\end{algorithmic}
\end{algorithm}
%

Our objective is to harness the computation capabilities in each area to perform state estimation \eqref{x_opt} and \eqref{variance_opt} online in a decentralized fashion, as shown on the right in Fig. \ref{fig.architecture}. Note that each area estimates {\it the global state} $\mathbf{v}$, rather than the portion that pertains to its local facilities. Since {\it step (1)} and {\it step (3)} in Algorithm \ref{ARGN_algorithm} are decoupled between different areas, their decentralized implementations are straightforward. Now we omit the time index $t$ and focus on solving {\it step (2)}
\begin{align}\label{centralized_SE}
	\widehat{\mathbf{v}} &= \arg\underset{\mathbf{v}\in\mathbb{V}}{\min}~
    	\sum_{i=1}^I\|\mathbf{\tilde{c}}_i-\mathbf{\tilde{f}}_i(\mathbf{v})\|^2,
\end{align}
where $\mathbf{\tilde{f}}_i(\mathbf{v})= \bdsb{\Gamma}_i^{-\frac{1}{2}}\mathbf{f}_i(\mathbf{v})$ and $\mathbf{\tilde{c}}_i=\bdsb{\Gamma}_i^{-\frac{1}{2}}\mathbf{c}_i$. Note that we propose to solve \eqref{centralized_SE} in a decentralized setting, where each area has a local estimate $\mathbf{v}_i^k$ that is in consensus with other areas $i'\neq i$ and converges to the global estimate $\widehat{\mathbf{v}}$.

Traditionally, state estimation solvers attempt to find the global minimum $\widehat{\mathbf{v}}$ numerically by Gauss-Newton (GN) algorithm iterations, whose updates are given by
\begin{align}
	\mathbf{v}_i^{k+1} &= P_{\mathbb{V}}\left[\mathbf{v}_i^k + \mathbf{d}_i^k\right],\quad
	\mathbf{d}_i^k = \mathbf{Q}^{-1}(\mathbf{v}_i^k)\mathbf{q}(\mathbf{v}_i^k),\label{central_descent1}
\end{align}
where $P_{\mathbb{V}}(\cdot)$ is a projection on the space $\mathbb{V}$, $\mathbf{q}(\mathbf{v}_i^k)$ and $\mathbf{Q}(\mathbf{v}_i^k)$ are scaled gradients and GN Hessian of the cost function
\begin{align}\label{rR}
	\mathbf{q}(\mathbf{v}_i^k) &= \frac{1}{I}\sum_{p=1}^{I}\mathbf{\tilde{F}}_p^T(\mathbf{v}_i^k)(\mathbf{\tilde{c}}_p-\mathbf{\tilde{f}}_p(\mathbf{v}_i^k))\\
	\mathbf{Q}(\mathbf{v}_i^k) &= \frac{1}{I}\sum_{p=1}^{I}\mathbf{\tilde{F}}_p^T(\mathbf{v}_i^k)\mathbf{\tilde{F}}_p(\mathbf{v}_i^k),\nonumber
\end{align}
with $\mathbf{\tilde{F}}_i(\mathbf{v}) =\bdsb{\Gamma}_i^{-\frac{1}{2}}\mathbf{T}_i{\mathrm{d}\mathbf{f}(\mathbf{v})}/{\mathrm{d}\mathbf{v}^T} =\bdsb{\Gamma}_i^{-\frac{1}{2}}\mathbf{T}_i\mathbf{F}(\mathbf{v})$ computed from $\mathbf{F}(\mathbf{v})$ in Appendix \ref{power_flow_equations_appendix}. However, each area knows only its own iterate $\mathbf{v}_i^k$ and partial measurements $\mathbf{\tilde{c}}_i$ and power flow equations $\mathbf{\tilde{f}}_i(\cdot)$ in \eqref{rR}, which makes it impossible to implement {\it step (2)} in a decentralized setting.


\section{Decentralized State Estimation and Tracking}\label{Decentralized_SE}
As discussed in Section \ref{ARSE_sec}, it is straightforward to decentralize the computations for {\it step (1)} and {\it step (3)} in Algorithm \ref{ARGN_algorithm}. The key enabling technique we propose for {\it step (2)},  is the Gossip-based Gauss-Newton (GGN) algorithm that we proposed in \cite{li2012convergence}. 
Next we describe the GGN algorithm to make the paper self-contained. We interchangeably use {\it area} and {\it agent} to refer to the entity communicating and performing the computation. There are two time scales in the GGN algorithm, one is the time for GN {\it update} denoted by the discrete time index ``$k$" and the other is the gossip {\it exchange} between every two GN updates denoted by another discrete time index ``$\ell$". All the agents have a clock that runs synchronously and determines the instants $t=\tau_k$ for the $k$-th GN update across the network. During the interval $t\in[\tau_k,\tau_{k+1})$, the agents communicate and exchange information with each other at random times $\tau_{k,\ell}\in[\tau_k,\tau_{k+1})$ over $\ell=1,\cdots,\ell_k$ interactions. In the following, we describe the local update model at each  agent in Section \ref{update_model} and introduce in Section \ref{gossip_model} the gossiping model between every two updates.

%

\subsection{Local Update Model}\label{update_model}
The idea behind our GGN algorithm is to take advantage of the structure of the centralized update \eqref{rR} as a sum of local terms and obtain an approximation of the sum by computing the scaled average via gossiping, which is interlaced with the optimization iterations. Therefore, we use the ``network average'' of different areas as surrogates of $\mathbf{q}(\mathbf{v}_i^k)$ and $\mathbf{Q}(\mathbf{v}_i^k)$, which can be obtained via gossiping
\begin{align}\label{Hh}
	\bar{\mathbf{h}}_k &=  \frac{1}{I}\sum_{i=1}^{I}\mathbf{\tilde{F}}_i^T(\mathbf{v}_i^k)\left(\mathbf{\tilde{c}}_i-\mathbf{\tilde{f}}_i(\mathbf{v}_i^k)\right)\\
	\bar{\mathbf{H}}_k &=  \frac{1}{I}\sum_{i=1}^{I}\mathbf{\tilde{F}}_i^T(\mathbf{v}_i^k)\mathbf{\tilde{F}}_i(\mathbf{v}_i^k).
\end{align}
Define local vector at the $i$-th agent for the $\ell$-th gossip
\begin{align}\label{information_vec}
	\bdsb{\mathcal{H}}_{k,i}(\ell)
	=
	\begin{bmatrix}
		\mathbf{h}_{k,i}(\ell)\\
		\mathrm{vec}\left[\mathbf{H}_{k,i}(\ell)\right]
	\end{bmatrix},
\end{align}
evolving from initial conditions
\begin{align}
	\mathbf{h}_{k,i}(0) &\triangleq\mathbf{\tilde{F}}_i^T(\mathbf{v}_i^k)(\mathbf{\tilde{c}}_i-\mathbf{\tilde{f}}_i(\mathbf{v}_i^k))\\
	\mathbf{H}_{k,i}(0) &\triangleq \mathbf{\tilde{F}}_i^T(\mathbf{v}_i^k)\mathbf{\tilde{F}}_i(\mathbf{v}_i^k).
\end{align}	
Clearly, the surrogates are the averages of the initial conditions
\begin{align}
    \bar{\mathbf{h}}_k &= \sum_{i=1}^{I}\mathbf{h}_{k,i}(0)/I,\quad
    \bar{\mathbf{H}}_k = \sum_{i=1}^{I}\mathbf{H}_{k,i}(0)/I.
\end{align}
To compute this average in the network, all agents exchange their information $\bdsb{\mathcal{H}}_{k,i}(\ell) \rightarrow \bdsb{\mathcal{H}}_{k,i}(\ell+1)$ using the communication model, more precisely the gossip exchange equation \eqref{gossip_exchange}, as described below in Section \ref{gossip_model}. Then after $\ell_k$ exchanges, the local GGN descent at the $(k+1)$-th update for the $i$-th agent is
\begin{align}\label{local_descent}
	\mathbf{v}_i^{k+1} &= P_{\mathbb{V}}\left[\mathbf{v}_i^k + \mathbf{d}_i^k(\ell_k)\right],\\
	\mathbf{d}_i^k(\ell_k)
	&=  \mathbf{H}_{k,i}^{-1}(\ell_k) \mathbf{h}_{k,i}(\ell_k).
\end{align}

Finally, the {\it Decentralized Adaptive Re-weighted State Estimation} (DARSE) scheme is described in Algorithm \ref{DARSE_algorithm}.

\begin{algorithm}[t]
\caption{DARSE Scheme}\label{DARSE_algorithm}
\begin{algorithmic}[1]
\STATE Predict outlier covariance at each agent $\bdsb{\Gamma}_i = \widehat{\mathbf{R}}_i[t-1]$.
\STATE All agents run the GGN Algorithm via network gossiping.
\STATE {\bf obtain} initial variables  $\mathbf{v}_i^0$ at all agents $i\in\mathcal{I}$.
\STATE {\bf set} $k=0.$
\REPEAT
\STATE {\bf set} $k=k+1.$
\STATE {\bf initialization:} Obtain $\bdsb{\mathcal{H}}_{k,i}(0)$ in \eqref{information_vec} at each agent $i$.
\STATE {\bf gossiping:} Each agent $i$ exchanges with neighbors under URE protocol for $1\leq \ell \leq \ell_k$ according to \eqref{gossip_exchange}.
\STATE {\bf local update:} Each agent $i$ updates according to \eqref{local_descent}.
\UNTIL $k=K$ or $\left\|\mathbf{v}_i^{k+1}-\mathbf{v}_i^k\right\|\leq \epsilon$ and set $\widehat{\mathbf{v}}_i = \mathbf{v}_i^K$.
\STATE Adjust covariance $\widehat{\mathbf{R}}_i[t] = \mathrm{diag}[\widehat{\epsilon}_{i,1}[t],\cdots,\widehat{\epsilon}_{i,M_i}[t]]$
\begin{align*}
	\widehat{\epsilon}_{i,m}[t] =  |c_{i,m}[t] - f_{i,m}(\widehat{\mathbf{v}}[t])|^2.
\end{align*}
\end{algorithmic}
\end{algorithm}

\subsection{Uncoordinated Random Exchange (URE) Protocol}\label{gossip_model}
In this section, we introduce the Uncoordinated Random Exchange (URE) protocol, a popular gossip algorithm studied in literature \cite{tsitsiklis1984problems,blondel2005convergence,nedic2009distributed}. For each exchange in the URE protocol, an agent $i$ wakes up and chooses a neighbor agent $j$ to communicate during $[\tau_k,\tau_{k+1})$. The communication network is thus a time-varying graph $\mathcal{G}_{k,\ell}=(\mathcal{I},\mathcal{M}_{k,\ell})$ during $[\tau_{k,\ell},\tau_{k,\ell+1})$ for every GN update $k$ and gossip exchange $\ell$. The node set $\mathcal{I}=\{1,\cdots,I\}$ contains each {\it agent} in different area, and the edge set $\{i,j\}\in\mathcal{M}_{k,\ell}$ is formed by the communication links for that particular gossip exchange, which can be characterized by the adjacency matrix $\mathbf{A}_k(\ell)=[A_{ij}^{(k,\ell)}]_{I\times I}$
\begin{align}
	A_{ij}^{(k,\ell)} =
	\begin{cases}
		1, &\{i,j\}\in\mathcal{M}_{k,\ell}\\
		0, &\mathrm{otherwise}
	\end{cases}.
\end{align}

\begin{condition}\label{connectivity_frequency}
The composite graph $\mathcal{G}_k=\{\mathcal{I}, \bigcup_{\ell'=\ell}^{\infty}\mathcal{M}_{k,\ell'}\}$ for the $k$-th update is connected for all $\ell \geq 0$ and there exists an integer $L\geq 1$ such that for every agent pair $\{i,j\}$ in the composite graph, we have for any $\ell\geq 0$
\begin{align}
	\{i,j\} \in \mathcal{M}_{k,\ell}\bigcup \mathcal{M}_{k,\ell+1} \bigcup \cdots \bigcup \mathcal{M}_{k,\ell+L-1}.
\end{align}
\end{condition}
\noindent The above assumption states that all agent pairs $\{i,j\}$ that communicate directly infinitely many times constitute a connected network $\mathcal{G}_k$, and furthermore, there exists an active link between any agent pair $\{i,j\}\in\mathcal{G}_k$ every $L$ consecutive time slots $[\tau_{k,\ell},\tau_{k,\ell+L-1}]\subseteq(\tau_k,\tau_{k+1})$ for any $\ell$.


The gossip exchanges are pairwise and local \cite{boyd2006randomized}, where agent $i$ combines the information from agent $j$ with a certain weight $\beta$. Define a weight matrix $\mathbf{W}_k(\ell)\triangleq [W_{ij}^{k}(\ell)]_{I\times I}$ as
\begin{align}
	W_{ij}^{k}(\ell)
	=
	\begin{cases}
		\beta, & j\neq i, j\in\mathcal{M}_{k,\ell}\\
		(1-\beta), & j=i\\
		0, & \mathrm{otherwise}
	\end{cases}
\end{align}
representing the weight associated to the edge $\{i,j\}$. Therefore, the weight matrix $\mathbf{W}_k(\ell)$ has the same sparsity pattern as the communication network graph $\mathbf{A}_k(\ell)$, and it is determined by the agent connectivity. Suppose agent $I_{k,\ell}$ wakes up at $\tau_{k,\ell}\in[\tau_k,\tau_{k+1})$ and $J_{k,\ell}$ is the node picked by node $I_{k,\ell}$ with probability $\gamma_{I_{k,\ell},J_{k,\ell}}$. The weight matrix is then
\begin{align}
	\mathbf{W}_k(\ell) = \mathbf{I} - \beta \left(\mathbf{e}_{I_{k,\ell}} + \mathbf{e}_{J_{k,\ell}}\right)\left(\mathbf{e}_{I_{k,\ell}} + \mathbf{e}_{J_{k,\ell}}\right)^T.
\end{align}
Finally, each agent $i=1,\cdots,I$ mixes its local information with neighbors as
\begin{align}\label{gossip_exchange}
	\bdsb{\mathcal{H}}_{k,i}(\ell+1) = W_{ii}^{k}(\ell) \bdsb{\mathcal{H}}_{k,i}(\ell) + \sum_{j\neq i} W_{ij}^{k}(\ell) \bdsb{\mathcal{H}}_{k,j}(\ell).
\end{align}

{\it Remark:} Note that the typical URE protocol is random and asynchronous, which may not satisfy Condition \ref{connectivity_frequency} due to link formations and failures. The simulations in \cite{yang2011transition} show that the overall delay from substations in a IEEE-14 bus system to the control center is around $2$ms with bandwidth $100$-$1000$ Mbits/s. Thus we bound the worst case hop delay by discounting it with the network diameter $2/7\approx 0.6$ms. We assume that the state estimation here is performed every $10$ seconds rather than today's periodicity (minutes) \cite{yang2011transition}. If information is stored with $64$-bits per entry, the data packets sent by each agent per exchange has $64(2N+4N^2)$-bits. For a power system with $N=118$ buses with a communication bandwidth $100$ Mbits/s, the maximum exchange that can be accomodated in $10$ seconds is $10\times 10^8/(0.6\times 10^{-3} \times 10^8 + 64(2\times 118+4\times 118^2))\approx 300$, which is sufficiently large to avoid violating Condition \ref{connectivity_frequency}. Much fewer exchanges were used in our simulations, but the algorithm still converges with good performance.

\section{Convergence and Performance Guarantees}
\begin{condition}\label{lipshitz}
First, we impose the following conditions:
\begin{itemize}
	\item[(1)] The state space $\mathbb{V}$ is closed and convex.
    \item[(2)] The maximum and minimum costs are bounded and finite
    			\begin{align}
				\epsilon_{\max} &= \underset{\mathbf{v}\in\mathbb{V}}{\max} ~\sum_{i=1}^I\|\mathbf{\tilde{c}}_i-\mathbf{\tilde{f}}_i(\mathbf{v})\|<\infty\\
				\epsilon_{\min} &= \underset{\mathbf{v}\in\mathbb{V}}{\min}
				~\sum_{i=1}^I\|\mathbf{\tilde{c}}_i-\mathbf{\tilde{f}}_i(\mathbf{v})\|<\infty.
			\end{align}
    \item[(3)] The maximum and minimum eigenvalues of the GN Hessian are non-zero and finite
			\begin{align}
				    \sigma_{\min} &= \underset{\mathbf{v}\in\mathbb{V}}{\min}~\sqrt{\lambda_{\min}\left(\sum_{i=1}^I\mathbf{\tilde{F}}_i^T(\mathbf{v})\mathbf{\tilde{F}}_i(\mathbf{v})\right)}>0\\
				    \sigma_{\max} &= \underset{\mathbf{v}\in\mathbb{V}}{\min}~\sqrt{\lambda_{\max}\left(\sum_{i=1}^I\mathbf{\tilde{F}}_i^T(\mathbf{v})\mathbf{\tilde{F}}_i(\mathbf{v})\right)}<\infty.
			\end{align}

\end{itemize}
\end{condition}
Condition 2-(1) is easily satisfied by setting standard voltage limits and conditions 2-(2) \& 2-(3) can be guaranteed if the power system is observable \cite{clements1990observability} and the measurement noise is finite. 
Next, we prove that the Jacobians $\{\mathbf{\tilde{F}}_i(\mathbf{v})\}_{i=1}^I$ satisfy the Lipschitz condition, which is important for the convergence of GGN algorithm in the DARSE scheme.

%
\begin{lem}\label{lem_Jacobian_lipschitz}
The Jacobian matrix $\mathbf{\tilde{F}}_i(\mathbf{v})$ satisfies the Lipschitz condition for all $i$ and arbitrary $\mathbf{v},\mathbf{v}'\in\mathbb{V}$
\begin{align}
	\left\|\mathbf{\tilde{F}}_i(\mathbf{v})-\mathbf{\tilde{F}}_i(\mathbf{v}')\right\|\leq  \omega \left\|\mathbf{v}-\mathbf{v}'\right\|,\quad \forall i=1,\cdots,I
\end{align}	
where $\omega$ is the Lipschitz constant given in \eqref{quadratic_form}.
\end{lem}
\begin{proof}
	See Appendix \ref{proof_lem_Jacobian_lipschitz}.
\end{proof}
\begin{corollary}\label{cor_Jacobian_lipschitz}
Given Condition \ref{lipshitz}, Lemma \ref{lem_Jacobian_lipschitz} and \cite[Theorem 12.4]{eriksson2004applied}, the following functions satisfy the Lipschitz conditions for all $\mathbf{v},\mathbf{v}'\in\mathbb{V}$
\begin{align*}
    \left\|\mathbf{\tilde{F}}_i^T(\mathbf{v})(\mathbf{\tilde{c}}_i-\mathbf{\tilde{f}}_i(\mathbf{v}))-\mathbf{\tilde{F}}_i^T(\mathbf{v}')(\mathbf{\tilde{c}}_i-\mathbf{\tilde{f}}_i(\mathbf{v}'))\right\|    &\leq \nu_{\delta}\left\|\mathbf{v}-\mathbf{v}'\right\|\\
    \left\|\mathbf{F}_i^T(\mathbf{v})\mathbf{F}_i(\mathbf{v})-\mathbf{F}_i^T(\mathbf{v}')\mathbf{F}_i(\mathbf{v}')\right\|    &\leq \nu_{\Delta}\left\|\mathbf{v}-\mathbf{v}'\right\|,
\end{align*}
where $\nu_{\delta} = \omega (\epsilon_{\max} + \sigma_{\max})$ and $\nu_{\Delta} = 2 \sigma_{\max} \omega$.
\end{corollary}

\subsection{Convergence Analysis}\label{convergence}

The GGN algorithm is initialized with $\mathbf{v}_i^0$ at each agent and continues until a stopping criterion is met. Since \eqref{centralized_SE} is NLLS problem that is non-convex similar to \eqref{ML_joint}, the iterate $\mathbf{v}_i^k$ may stop at any fixed point $\mathbf{v}^\star$ of \eqref{central_descent1} satisfying the first order condition similar to \eqref{x_opt_1}
\begin{align}\label{fixed_point}
	 \sum_{i=1}^I\mathbf{\tilde{F}}_i^T(\mathbf{v}^\star)\left(\mathbf{\tilde{c}}_i-\mathbf{\tilde{f}}_i(\mathbf{v}^\star)\right) = \mathbf{0}.
\end{align}
Clearly, the ML estimate $\widehat{\mathbf{v}}$ in \eqref{centralized_SE} is one of the fixed points and it is desirable to have the algorithm converge to this point. After we have proven Lemma \ref{lem_Jacobian_lipschitz} and Corollary \ref{cor_Jacobian_lipschitz} for power systems, the analysis presented in our paper \cite{li2012convergence} is particularly useful here. In the following, we impose a condition on the gossip exchange according to our analysis in \cite{li2012convergence}.
\begin{condition}\label{cond_gossip_exchange}
\cite{li2012convergence} Denote the minimum exchange as $\ell_\star = \min_k~\{\ell_k\}$. Let $\eta=\min\{\beta,1-\beta\}$ be the minimum non-zero entry in the weight matrix and let the sequence of exchanges $\{\ell_k\}_{k=0}^{\infty}$ satisfy\footnote{A simple choice is $\ell_0=\ell_\star$ and $\ell_k=\ell_{k-1}+1$, then $\lambda_\infty=1/(1-\lambda_\eta)$.}
\begin{align*}
    \lambda_\infty
    &\triangleq\underset{K\rightarrow \infty}{\lim}~\sum_{k=0}^{K} \lambda_\eta^{(\ell_k-\ell_\star)}< \infty\\
    \lambda_\eta &\triangleq (1-\eta^{IL})^{1/IL}.
\end{align*}
For any $\xi\in(0,1/2)$, the minimum exchange $\ell_\star$ is chosen as
\begin{align}\label{denominator}
	\ell_\star &= \left\lceil\log\left(\frac{\xi}{4D}\right)/\log\lambda_\eta\right\rceil\\
    D &\triangleq CC_2 (\nu\lambda_\infty C_1C_2 + 1)
\end{align}
where $\nu=\max\{\nu_\delta,\nu_\Delta\}$ and
\begin{align}\label{C_infty}
    C &\triangleq 2I\sigma_{\max}\sqrt{I(\epsilon_{\max}^2 + N\sigma_{\max}^2)} \left(\frac{1+\eta^{-IL}}{1-\eta^{IL}}\right)\\
    C_1 &\triangleq 2\left(1+\frac{\sigma_{\max}\epsilon_{\max}}{\sigma_{\min}^2}\right),\quad
    C_2=\frac{I}{\sigma_{\min}^2}
\end{align}
with $\epsilon_{\max}$, $\sigma_{\min}$ and $\sigma_{\max}$ given by Condition \ref{lipshitz}.
\end{condition}
The resultant theorem from \cite{li2012convergence} can be re-stated as follows.

    \begin{thm}\label{theorem_convergence}
\cite[Lemma 1, 2, 3 \& Theorem 1]{li2012convergence} Based on Lemma \ref{lem_Jacobian_lipschitz} and Condition \ref{connectivity_frequency}, \ref{lipshitz} \& \ref{cond_gossip_exchange}, then
the discrepancy between the local update \eqref{local_descent} and the exact update \eqref{central_descent1} is bounded for all $i$ and $k$
\begin{align}
	\left\|\mathbf{d}_i^k(\ell_k) - \mathbf{d}_i^k\right\|
	&\leq \kappa,
    \quad
    \kappa = 4C_1 D\lambda_\eta^{(\ell_\star+1)}
\end{align}
and the error between $\mathbf{v}_i^k$ generated by \eqref{local_descent} and the ML estimate $\widehat{\mathbf{v}}$ defined in \eqref{fixed_point} satisfies
\begin{align}\label{distributed_error_recursion}
	\left\|\mathbf{v}_i^{k+1}-\widehat{\mathbf{v}}\right\|
	&\leq
	T_1\left\|\mathbf{v}_i^k-\widehat{\mathbf{v}}\right\|^2
	 +T_2\left\|\mathbf{v}_i^k-\widehat{\mathbf{v}}\right\|+\kappa,
\end{align}
where $T_1 \triangleq {\omega}/{2\sigma_{\min}}$ and $T_2 \triangleq {\sqrt{2}\omega\epsilon_{\min}}/{\sigma_{\min}^2}$. Assuming that $\sqrt{2}\omega \epsilon_{\min} < 3\sigma_{\min}^2$, and that $\kappa \ll  {(1-T_2)^2}/{4 T_1}$, then for any $\mathbf{v}_i^0$ satisfying
\begin{align}\label{initial_condition_requirement}
	\left\|\mathbf{v}_i^0-\widehat{\mathbf{v}}\right\|< \frac{2\sigma_{\min}}{\omega} - \kappa,
\end{align}
the asymptotic estimation error can be bounded as
\begin{align}
	\underset{k\rightarrow\infty}{\limsup} \left\|\mathbf{v}_i^{k+1}-\widehat{\mathbf{v}}\right\| \leq \kappa.
\end{align}	
\end{thm}

Theorem \ref{theorem_convergence} implies that all the agents will converge to an arbitrarily small neighborhood of the ML estimate $\widehat{\mathbf{v}}$ if each area is initialized with a value that is sufficiently close to the ML estimate. Our simulations show that one can converge to the ML estimate from quite inaccurate initial points, even violating \eqref{initial_condition_requirement}. Also, Condition \ref{cond_gossip_exchange} is imposed on the number of exchanges $\ell_k$  to control the convergence rate (i.e., determined by $T_1$ and $T_2$) and lower the numerical error $\kappa$. This condition is influenced by many factors, such as the number of agents. More specifically, the number of measurements, measurement type and measurement location can in fact greatly affect the parameters $\{\epsilon_{\min},\epsilon_{\max},\sigma_{\min},\sigma_{\max}\}$ in Condition \ref{lipshitz}, which in turn influences $T_1$, $T_2$ and $\ell_\star$ that determine the convergence rate of the algorithm.

\subsection{Performance Guarantee by PMU Initialization}
Theorem \ref{theorem_convergence} suggests that if the GGN iterate $\mathbf{v}_i^k$ is initialized in a certain neighborhood of $\widehat{\mathbf{v}}$, it converges to $\widehat{\mathbf{v}}$ with an error $\kappa$ resulting from gossiping. This means that a good initialization $\mathbf{v}_i^0$ around the the ML estimate $\widehat{\mathbf{v}}$ is important. In fact, an effective initializer is the re-scaled average of the voltage measurements $\mathbf{c}_{i,\mathcal{V}}$ of all areas because it measures the state directly.
 Next we propose a heuristic initialization scheme, which is shown to converge numerically.

\subsubsection{Centralized PMU Initialization $\mathbf{v}^0$}
As new measurements become available, a reasonable approach to initialize the state estimates  is to use a combination of the previous state estimate for the buses where there are no PMU installed and the direct state measurements given by the PMU when available. The mathematical expression for this choice stated below helps describing and motivating the decentralized initialization scheme that follows.

The PMU data vector $\mathbf{z}_{\mathcal{V}}[t]$ in \eqref{meas-model_all}, records a portion of the state. By permuting the entries of $\mathbf{c}_{i,\mathcal{V}}[t]$ with the matrix $\mathbf{T}_{i,\mathcal{V}}$, we can relate \eqref{meas-model_all} and the decentralized model \eqref{meas-model} as follows
\begin{align}\label{T_i_Z}
	\mathbf{T}_{i,\mathcal{V}}^T \mathbf{c}_{i,\mathcal{V}}[t]
	&= \mathbf{T}_{i,\mathcal{V}}^T\mathbf{T}_{i,\mathcal{V}}\mathbf{z}_{\mathcal{V}}[t].
\end{align}
where, by virtue of of $\mathbf{T}_{i,\mathcal{V}}$ structure, the matrix $\mathbf{I}_{i,\mathcal{V}} = \mathbf{T}_{i,\mathcal{V}}^T\mathbf{T}_{i,\mathcal{V}}$ is a masked identity matrix with non-zero entries on locations that are PMU-instrumented. Letting $\mathbf{I}_{\mathcal{V}}=\sum_{i=1}^I\mathbf{I}_{i,\mathcal{V}}$, we have $\mathbf{I}_{\mathcal{V}}\mathbf{z}_{\mathcal{V}}[t] = \sum_{i=1}^{I} \mathbf{T}_{i,\mathcal{V}}^T \mathbf{c}_{i,\mathcal{V}}[t]$, Therefore, the centralized initializer that merges PMUs measurements and outdated estimates can be written as
\begin{align}\label{x0}
	\mathbf{v}^0  = \mathbf{I}_{\mathcal{V}}\mathbf{z}_{\mathcal{V}}[t] + (\mathbf{I}-\mathbf{I}_{\mathcal{V}})\mathbf{s}_{\mathcal{V}}[t],
\end{align}
where $\mathbf{s}_{\mathcal{V}}[t]$ can be chosen arbitrarily (e.g. a stale estimate). It is of great interest to investigate the placement of PMU devices so that certain metrics are optimized, such as observability \cite{clements1990observability}, state estimation accuracy \cite{kekatos2011convex,li2011phasor} and mutual information \cite{li2012information}. This issue is not pursued here and the numerical results are based on an arbitrary choice.

\subsubsection{Decentralized Initializer $\mathbf{v}_i^0$ via Gossiping}
The ``exact central initializer" in \eqref{x0} can be re-written as
\begin{align}\label{exact_central_initializer}
	\mathbf{v}^0=\sum_{i=1}^{I} \mathbf{T}_{i,\mathcal{V}}^T \mathbf{c}_{i,\mathcal{V}}[t] + \left(\mathbf{I}-\sum_{i=1}^{I} \mathbf{T}_{i,\mathcal{V}}^T\mathbf{T}_{i,\mathcal{V}}\right)\mathbf{s}_{\mathcal{V}}[t].
\end{align}
However, $\sum_{i=1}^I\mathbf{T}_{i,\mathcal{V}}^T\mathbf{c}_{i,\mathcal{V}}[t]$ is the aggregated PMU measurements, while the $i$-th agent can only access its local measurements $\mathbf{T}_{i,\mathcal{V}}^T\mathbf{c}_{i,\mathcal{V}}[t]$. Since $\sum_{i=1}^I\mathbf{T}_{i,\mathcal{V}}^T\mathbf{c}_{i,\mathcal{V}}[t]$ is written as a sum, it can be obtained via gossiping with the initial state
\begin{align}\label{mu_i_0}
	\bdsb{\mathcal{V}}_i(0) = \mathbf{T}_{i,\mathcal{V}}^T\mathbf{c}_{i,\mathcal{V}}[t],
\end{align}
where the agents proceed the exchange with the URE protocol $\bdsb{\mathcal{V}}_i(\ell)\rightarrow\bdsb{\mathcal{V}}_i(\ell+1)$. Finally, the decentralized initializer is
\begin{align}\label{decentralized_initializer}
	\mathbf{v}_i^0
	= I\bdsb{\mathcal{V}}_i(\ell) + (\mathbf{I}-\mathbf{I}_{\mathcal{V}})\mathbf{s}_{\mathcal{V}}[t],
\end{align}
where $\mathbf{I}_{\mathcal{V}} = \mathrm{sgn}[\bdsb{\mathcal{V}}_i(\ell)]$ and $\mathrm{sgn}[\cdot]$ is the sign function $\mathrm{sgn}[v]=1$ for $v\neq 0$ and $0$ otherwise. If $\ell$ is sufficiently large, $\mathbf{v}_i^0$ converges to the centralized initializer in \eqref{x0}.

\section{Numerical Results}\label{numerical}

\nocite{UK_grid}

In this section, we illustrate the Mean Square Error (MSE) performance of the DARSE scheme. Given the distributed estimate in each area $\{\widehat{V}_{i,n}^{(k)}\}_{n=1}^N$ at each GN update, the MSE with respect to the voltage magnitude and voltage phase at the $i$-th site is
\begin{align*}
    \mathrm{MSE}_{V,i}^{(k)}  &= \sum_{n=1}^{N}(|\bar{V}_n|-|\widehat{V}_{i,n}^{(k)}|)^2,\\
    \mathrm{MSE}_{\Theta,i}^{(k)}  &= \sum_{n=1}^{N}(\angle\bar{V}_n-\angle\widehat{V}_{i,n}^{(k)})^2.
\end{align*}

In particular, the metric used in our comparisons are the cost in \eqref{centralized_SE},  $\mathrm{Val}_k = \sum_{i=1}^{I} \|\mathbf{\tilde{c}}_i[t]-\mathbf{\tilde{f}}_i(\mathbf{v}_i^k[t])\|^2$ evaluated using the decentralized estimates at each update, and the global $\mathrm{MSE}_V^{(k)} = \sum_{i=1}^I  \mathrm{MSE}_{V,i}^{(k)}$ and $\mathrm{MSE}_\Theta^{(k)} = \sum_{i=1}^I  \mathrm{MSE}_{\Theta,i}^{(k)}$.

In the simulations we used MATPOWER 4.0 test case IEEE-118 ($N=118$) system. We take the load profile from the UK National Grid load curve from \cite{UK_grid} and scale the base load from MATPOWER on load buses. Then we run the Optimal Power Flow (OPF) program to determine the generation dispatch over this period. This gives us the true state $\bar{\mathbf{v}}[t]$ and all the power quantities $\mathbf{f}(\bar{\mathbf{v}}[t])$ over this time horizon, which are all expressed in per unit (p.u.) values. The sensor observations are generated by adding independent errors $r_{i,m}[t]\sim \mathcal{N}(0,\sigma^2)$ with $\sigma=10^{-3}$. We divide the system into $I=10$ areas where $9$ areas has $12$ buses in each area and $1$ area has $10$ buses, all chosen at random from $1$ to $118$. For each area, we randomly choose $50\%$ of all the available measurements and particularly exploit the $36$ PMU measurements in Area $1$, $2$ and $3$. The optimization of PMU selection is beyond the scope of this paper and hence not pursued here. In the first snapshot where there is no previous state estimate, we choose the flat profile $\mathbf{s}_{\mathcal{V}}[t]=[\mathbf{1}_N^T ~\mathbf{0}_N^T]^T$.

\subsection{Comparison with Diffusion Algorithms without Bad Data}
In this subsection, we evaluate the overall performance of the DARSE scheme against existing network diffusion algorithms \cite{xie2012fully} and its extension to adaptive processing in \cite{kar2008distributed} over $3$ snapshots of measurements. However, the communication protocol in \cite{xie2012fully,kar2008distributed} requires agents to exchange information synchronously, and thus the URE protocol mentioned in Section \ref{gossip_model} does not fit the context. To make a fair comparison in terms of communication costs and accuracy, we modify the URE protocol for this particular comparison to a synchronous deterministic exchange protocol, where communication links exist between every two agents $\{i,j\}\in\mathcal{M}$ for $\forall i,j\in\mathcal{I}$, giving an adjacency matrix $\mathbf{A}= \mathbf{1}_I\mathbf{1}_I^T-\mathbf{I}$. The weight matrix is doubly stochastic in both cases constructed according to the Laplacian $\mathbf{L}= \mathrm{diag}(\mathbf{A}\mathbf{1}_I)-\mathbf{A}$ as $\mathbf{W}= \mathbf{I}_I - w \mathbf{L}$ with $w=\alpha/\max(\mathbf{A}\mathbf{1}_I)$ and $\alpha = 0.03$. For simplicity, we also do not simulate measurements that have bad data in this particular comparison.  The step-sizes for the approach in \cite{xie2012fully,kar2008distributed} are chosen as $\alpha_{\rm diff} = 0.01\ell^{-1},0.3\ell^{-1},0.5\ell^{-1},\ell^{-1}$, where $\ell$ is the gossip exchange index.

The diffusion algorithm proceeds at each exchange $\ell$, while the DARSE runs $\ell_k=\ell_0=\ell_\star=10$ exchanges for each update. Thus, the comparison is made on the same time scale by counting the number of gossip exchanges. Because we use $\ell_\star=10$ gossip exchanges between every two descent updates $k=1,\cdots, 20$, we have a total number of $200$ exchanges per snapshot. It can be seen from Fig. \ref{fig.diff_obj_dynamic} to \ref{fig.diff_Va_dynamic} that the DARSE scheme converges to the ML estimate after $k=15$ updates (i.e., $150$ message exchanges) for every snapshot and it tracks the state estimate accurately when new measurements stream in. The spikes observed in the plots are caused by the new measurements. Since the number of gossip exchanges is limited, the diffusion algorithm in \cite{xie2012fully} and \cite{kar2008distributed} suffers from slow convergence and fails to track the state accurately.

\begin{figure}[!t]
\begin{center}
{\subfigure[][$\mathrm{Val}_k$]{\resizebox{0.43\textwidth}{!}{\includegraphics{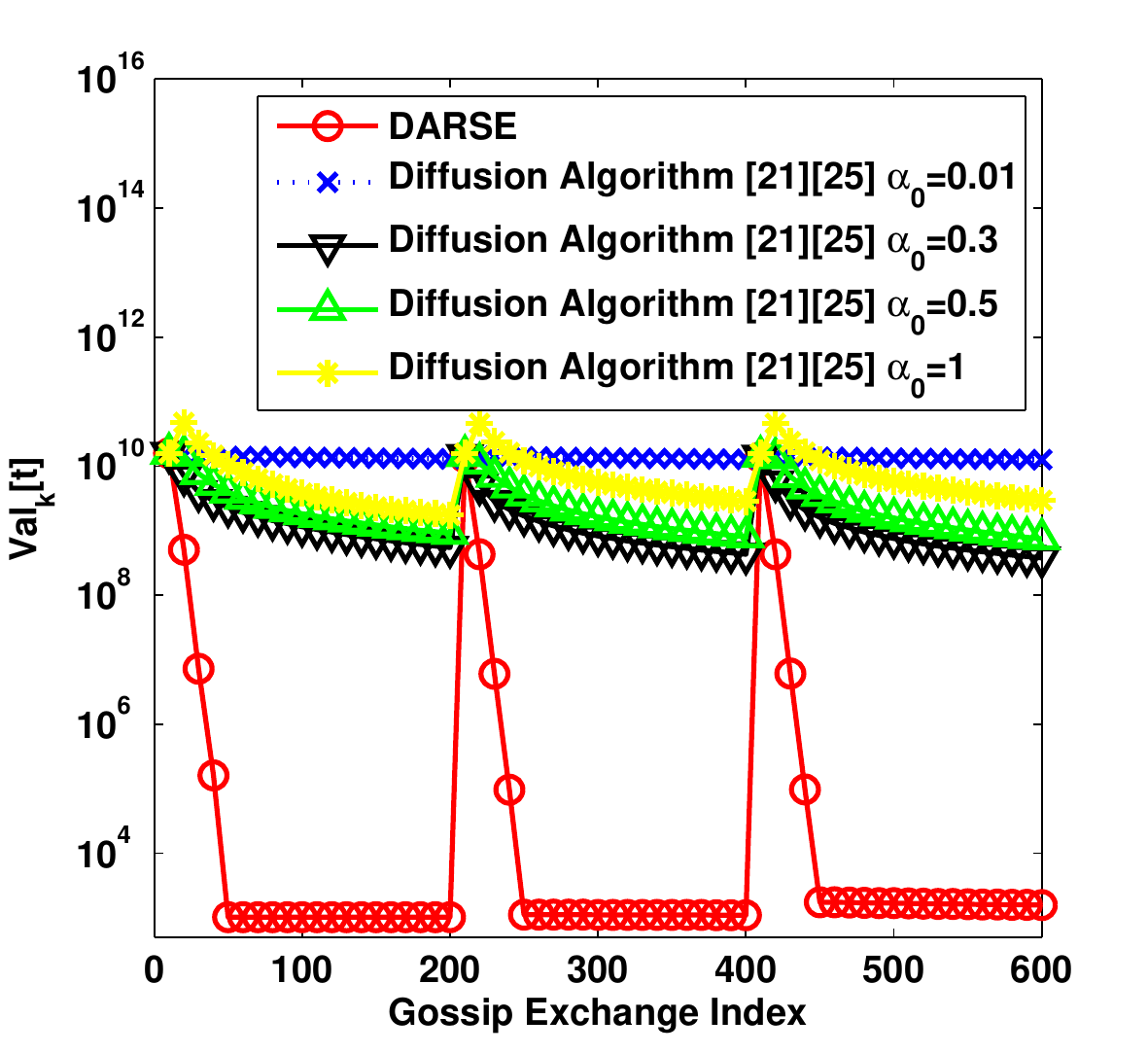}\label{fig.diff_obj_dynamic}}}}\\
\vspace{-0.2cm}
{\subfigure[][$\mathrm{MSE}_V^{(k)}$]{\resizebox{0.43\textwidth}{!}{\includegraphics{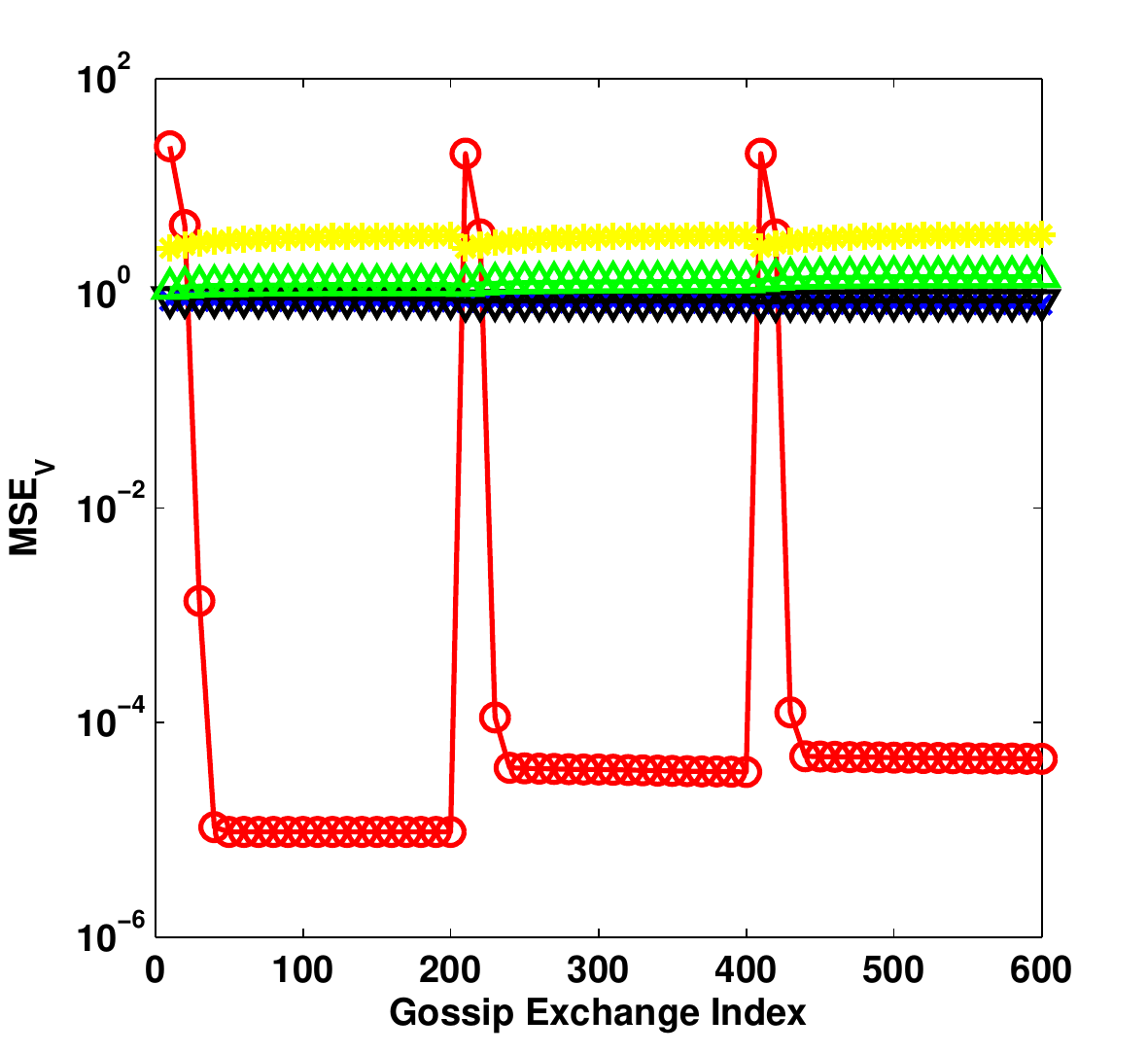}\label{fig.diff_Vm_dynamic}}}}\\
\vspace{-0.2cm}
{\subfigure[][$\mathrm{MSE}_\Theta^{(k)}$]{\resizebox{0.43\textwidth}{!}{\includegraphics{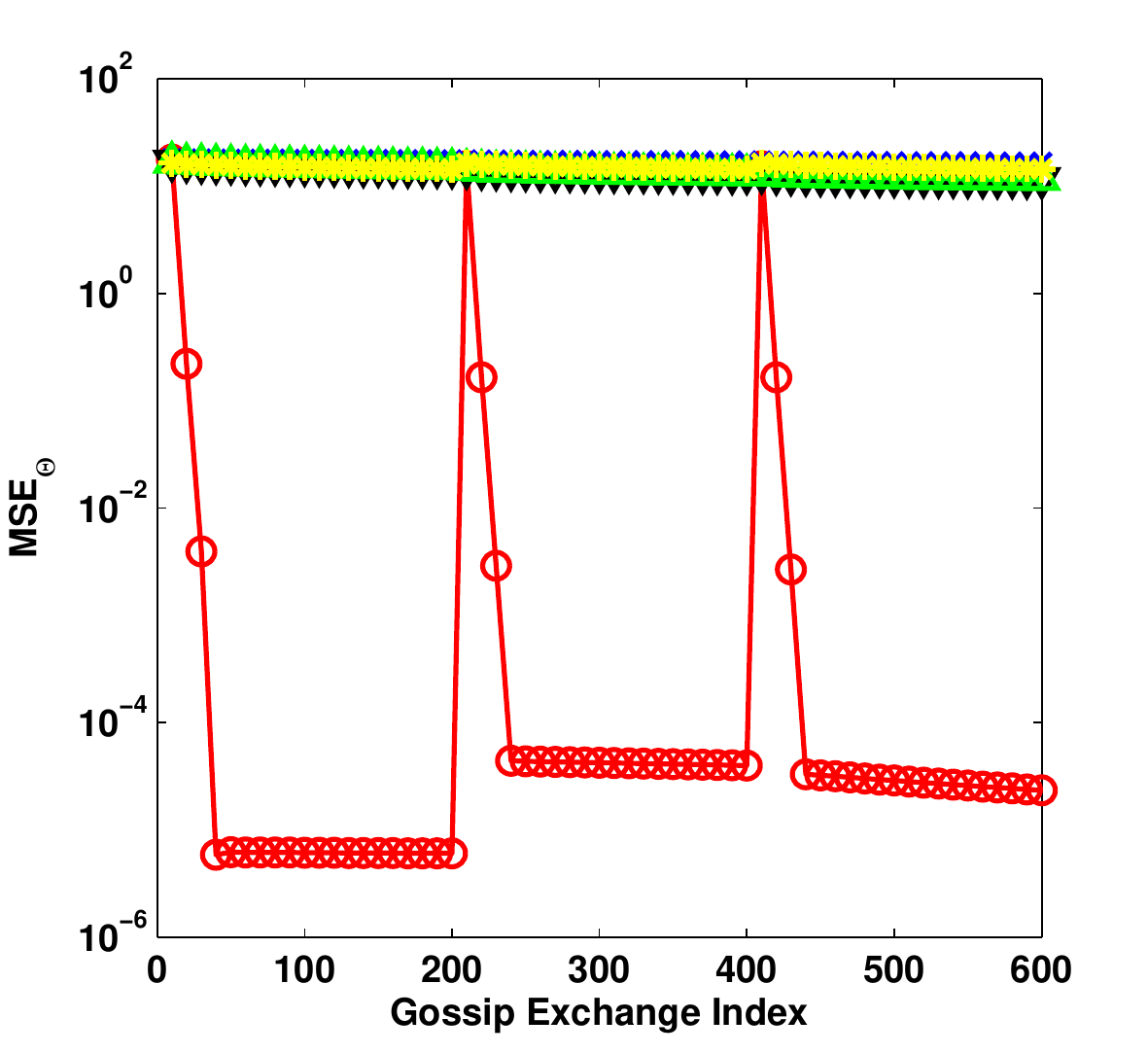}\label{fig.diff_Va_dynamic}}}}
\end{center}
\vspace{-0.4cm}
\caption{Comparison between DARSE and diffusion algorithms in \cite{xie2012fully,kar2008distributed} using $\ell_\star=3$ exchanges for every update.}
\vspace{-0.5cm}
\end{figure}

The reason for the fast convergence of the GGN algorithm is intrinsically that the GGN algorithm achieves the convergence rate of the GN algorithm, which converges quadratically when the noise level is low, while the diffusion algorithm is a first order sub-gradient method that converges sub-linearly. The fast convergence is partly due to the difference in computation complexity, where our algorithm is dominated by the matrix inversion on the order of $\mathcal{O}(N^3)$, while the diffusion algorithm scales like $\mathcal{O}(N)$. This requires the local processor to have the capability to maintain such computations on time for each exchange. There is literature on reducing the computation cost for matrix inversions, but this is beyond the scope of this paper.

\subsection{Comparison with Centralized GN Approach}
In this simulation, we added random outliers errors with variances $\epsilon_{i,m}[t]=100\sigma^2$ on $25$ randomly selected measurements for $6$ measurement snapshots. We examine the MSE performance of the DARSE scheme where,  in each snapshot $t$, each agent talks to another agent $20$ times on average during the interval $[\tau_k,\tau_{k+1})$ for all $k=1,\cdots,20$. In this case, the network communication volume is on the order of the network diameter $\mathcal{O}(N)$, which implies the number of transmissions in the centralized scheme as if the local measurements are relayed and routed through the entire network.
Furthermore, we examine the performance of the DARSE for cases with random link failures, where any established link $\{i,j\}\in\mathcal{M}$ fails with probability $p=0.1$ independently. It is clear that this communication model with link failures may not satisfy Condition \ref{connectivity_frequency}, \ref{lipshitz} \& \ref{cond_gossip_exchange}, but the numerical results show that our approach is robust and degrades gracefully.

To demonstrate the effectiveness of the DARSE scheme with bad data, we compare it with the centralized GN procedure with and without bad data, where the situation without bad data serves as the ultimate benchmark. Clearly, it can be seen from Fig. \ref{fig.bad_data_for_both} to \ref{fig.bad_data_Va} that DARSE sometimes has a certain performance loss compared with the centralized GN without bad data. Sometimes DARSE outperforms the centralized GN algorithm because the re-weighting numerically leads to certain improvement, but this is due to the fact that the measurements with greater variance influence less the state estimates rather than an intrinsic behavior of the algorithm. On the other hand, when bad data are present, the DARSE scheme outperforms significantly the centralized GN approach without re-weighting, thanks to the bad data suppression.

\begin{figure}[!t]
\begin{center}
{\subfigure[][$\mathrm{Val}_k$]{\resizebox{0.42\textwidth}{!}{\includegraphics{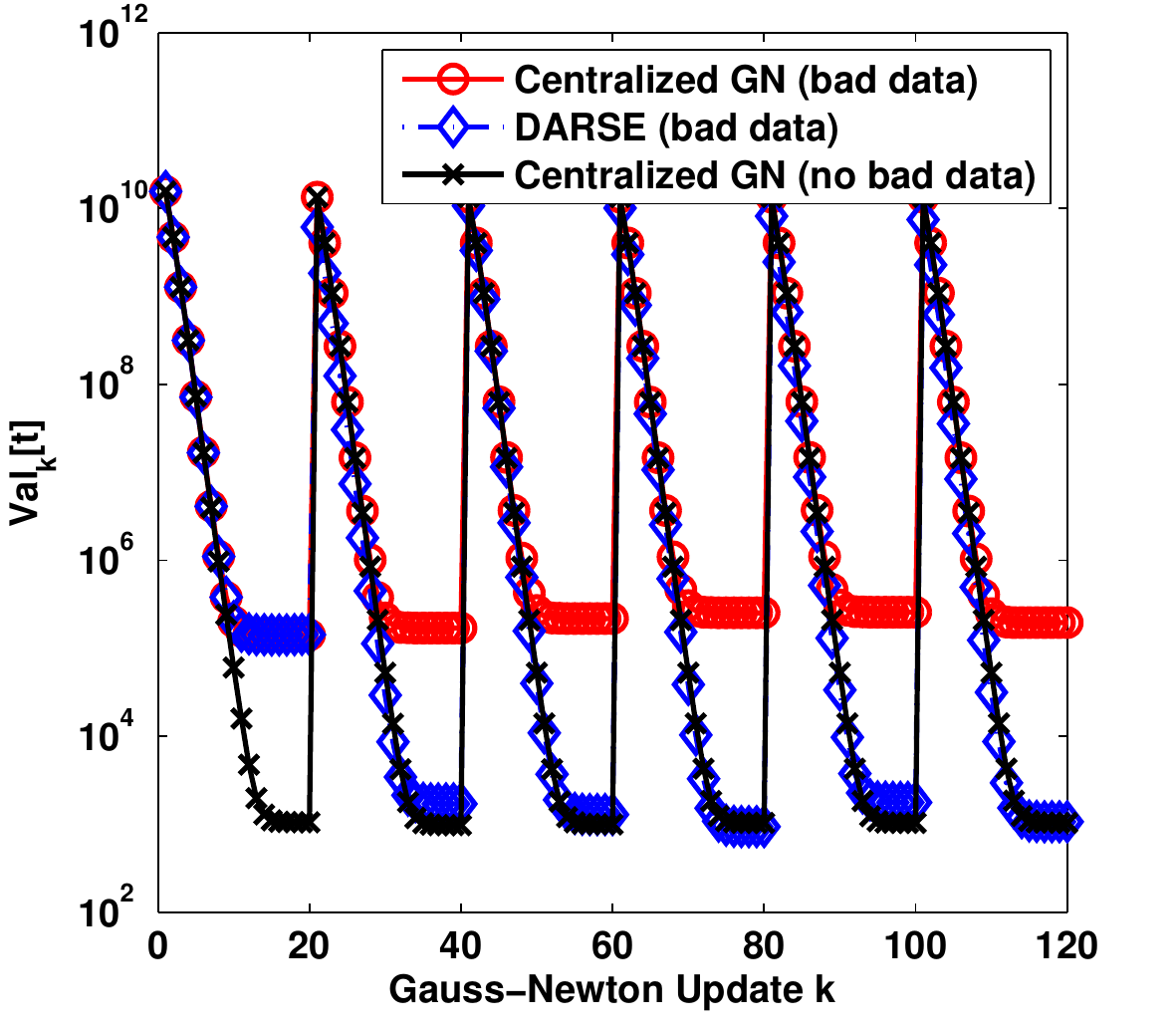}\label{fig.bad_data_for_both}}}}\\
{\subfigure[][$\mathrm{MSE}_V^{(k)}$]{\resizebox{0.42\textwidth}{!}{\includegraphics{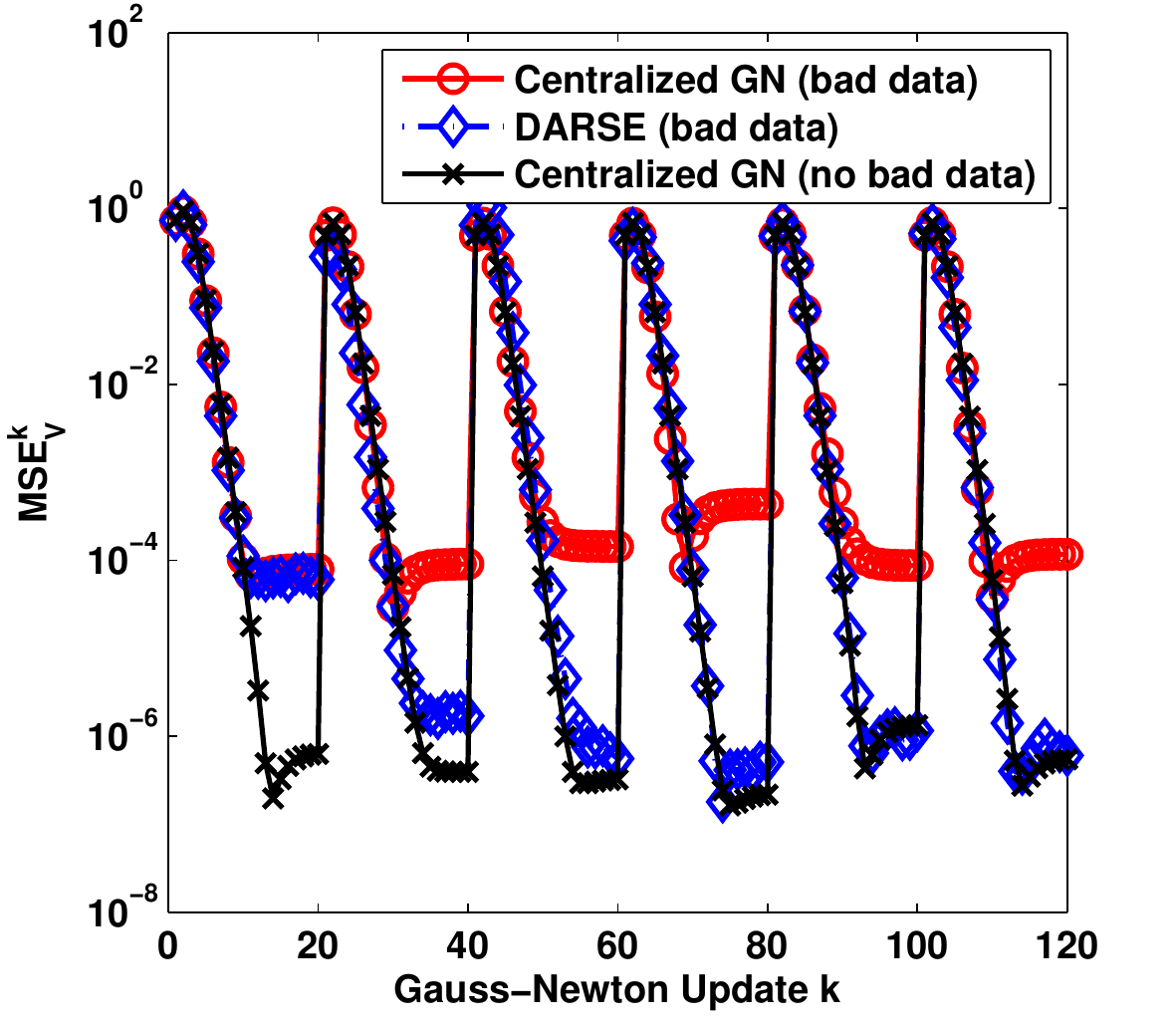}\label{fig.bad_data_Vm}}}}\\
{\subfigure[][$\mathrm{MSE}_\Theta^{(k)}$]{\resizebox{0.42\textwidth}{!}{\includegraphics{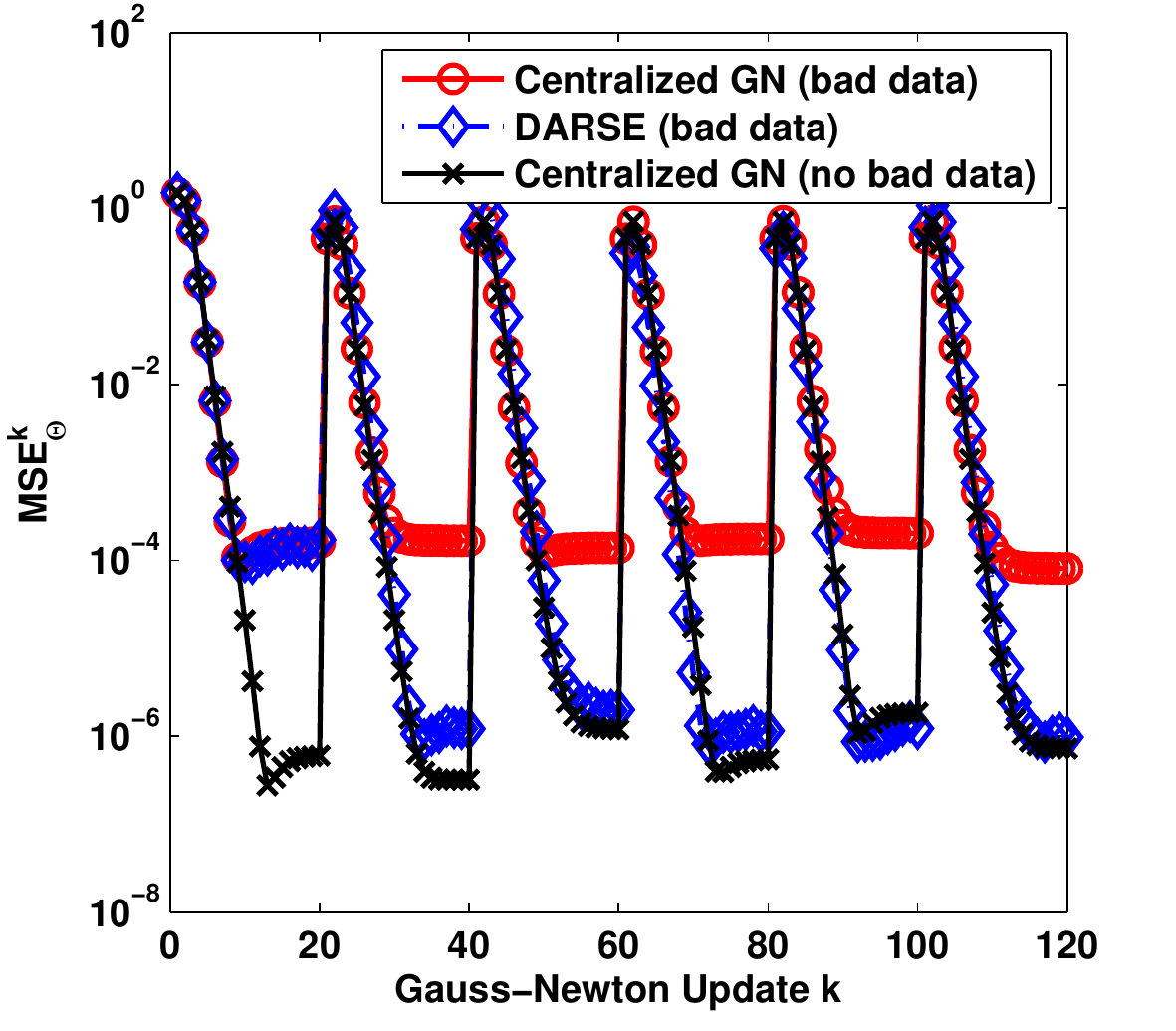}\label{fig.bad_data_Va}}}}
\end{center}
\caption{Performance of DARSE against centralized GN with and without bad data in terms of $\mathrm{Val}_k$, $\mathrm{MSE}_V^{(k)}$ and $\mathrm{MSE}_\Theta^{(k)}$}
\end{figure}

\section{Conclusions}
In this paper, we propose a DARSE scheme for hybrid power system state estimation integrating seamlessly WAMS an SCADA measurement system, which adaptively estimates the global state vector along with an updated noise covariance. The numerical results show that the DARSE scheme is able to deliver accurate estimates of the entire state vector at each distributed area, even in the presence of bad data and random communication link failures.

\appendices
\section{Power Flow Equations and Jacobian Matrix}\label{power_flow_equations_appendix}

Each line is characterized by the admittance matrix $\mathbf{Y}=[-Y_{nm}]_{N\times N}$, which includes line admittances $Y_{nm}=G_{nm}+\mathrm{i}B_{nm}$, $\{n,m\}\in\mathcal{E}$, shunt admittances $\bar{Y}_{nm}=\bar{G}_{nm}+\mathrm{i}\bar{B}_{nm}$ in the $\Pi$-model of line $\{n,m\}\in\mathcal{E}$, and self-admittance $Y_{nn}= - \sum_{m\neq n} (\bar{Y}_{nm} + Y_{nm})$. Using the canonical basis $\mathbf{e}_n=[0,\cdots,1,\cdots,0]^T$ and $\mathbf{Y}$, we define the following
\begin{align*}
	\mathbf{Y}_n &\triangleq \mathbf{e}_n\mathbf{e}_n^T\mathbf{Y},
	\quad
	\mathbf{Y}_{nm} &\triangleq (Y_{nm}+\bar{Y}_{nm})\mathbf{e}_n\mathbf{e}_n^T - Y_{nm}\mathbf{e}_n\mathbf{e}_m^T.
\end{align*}
Letting $\mathbf{G}_n = \Re\{\mathbf{Y}_n\}$, $\mathbf{B}_n=\Im\{\mathbf{Y}_n\}$, $\mathbf{G}_{nm}=\Re\{\mathbf{Y}_{nm}\}$ and $\mathbf{B}_{nm}=\Im\{\mathbf{Y}_{nm}\}$, we define the following matrices
\begin{align*}
    \mathbf{N}_{P,n}
    &\triangleq
    \begin{bmatrix}
        \mathbf{G}_n & - \mathbf{B}_n \\
        \mathbf{B}_n & \mathbf{G}_n
    \end{bmatrix}~~~~~~
    \mathbf{N}_{Q,n}
    \triangleq
    -
    \begin{bmatrix}
        \mathbf{B}_n & \mathbf{G}_n\\
        -\mathbf{G}_n & \mathbf{B}_n
    \end{bmatrix}\\
    \mathbf{E}_{P,nm}
    &\triangleq
    \begin{bmatrix}
        \mathbf{G}_{nm} &  - \mathbf{B}_{nm} \\
        \mathbf{B}_{nm}  & \mathbf{G}_{nm}
    \end{bmatrix}~
    \mathbf{E}_{Q,nm}
    \triangleq
    -
    \begin{bmatrix}
        \mathbf{B}_{nm}  & \mathbf{G}_{nm}\\
        -\mathbf{G}_{nm} & \mathbf{B}_{nm}
    \end{bmatrix}\\
    \mathbf{C}_{I,nm}
    &\triangleq
    \begin{bmatrix}
        \mathbf{G}_{nm}  & \mathbf{0}\\
        \mathbf{0}  & -\mathbf{B}_{nm}
    \end{bmatrix}~~
    \mathbf{C}_{J,nm}
    \triangleq
    ~
    \begin{bmatrix}
        \mathbf{B}_{nm}  & \mathbf{0}\\
        \mathbf{0}  & \mathbf{G}_{nm}
    \end{bmatrix}.
\end{align*}

State estimation mainly uses power injection and flow measurements from SCADA systems or, if available, PMU measurements from WAMS. Traditional SCADA systems aggregate data from the so called Remote Terminal Units (RTU), refreshing the data every $T_{\textrm{\tiny SCADA}} = 2$ to $5$ seconds and collects active/reactive injection $(P_n,Q_n)$ at bus $n$ and flow $(P_{nm},Q_{nm})$ at bus $n$ on line $\{n,m\}$
		\begin{align}
			P_n
			&=  \mathbf{v}^T\mathbf{N}_{P,n}\mathbf{v},\quad
			P_{nm}
			= \mathbf{v}^T\mathbf{E}_{P,nm}\mathbf{v}\label{power_flow_eq_SCADA_power_inj}\\
			Q_n
			&=  \mathbf{v}^T\mathbf{N}_{Q,n}\mathbf{v},\quad
			Q_{nm}
			= \mathbf{v}^T\mathbf{E}_{Q,nm}\mathbf{v},\label{power_flow_eq_SCADA_power_flow}
		\end{align}
and stack them in the power flow equations
		\begin{align}
			\mathbf{f}_{\mathcal{I}}(\mathbf{v}) &= [\cdots,P_n, \cdots, \cdots, Q_n, \cdots]^T\\
			\mathbf{f}_{\mathcal{F}}(\mathbf{v}) &= [\cdots,P_{nm},\cdots,\cdots,Q_{nm}, \cdots]^T.
		\end{align}
The WAMS generate data at a much faster pace compared to SCADA systems, with $T_{\textrm{\tiny PMU}} = 1/120$ to $1/30$ second. The PMU data are gathered at Phasor Data Concentrators (PDC), which collects the voltage $(\Re\{V_n\},\Im\{V_n\})$ at bus $n$ and the current $(I_{nm},J_{nm})$ on line $\{n,m\}$ measured at bus $n$
		\begin{align}\label{power_flow_eq_PMU}
			I_{nm}
			&= \left(\mathbf{1}_2\otimes\mathbf{e}_n\right)^T\mathbf{C}_{I,nm}\mathbf{v}\\
			J_{nm}
			&= \left(\mathbf{1}_2\otimes\mathbf{e}_n\right)^T\mathbf{C}_{J,nm}\mathbf{v},
		\end{align}
where $\otimes$ is the Kronecker product, and stacks them as
\begin{align}
	\mathbf{f}_{\mathcal{V}}(\mathbf{v}) = \mathbf{v},~\mathbf{f}_{\mathcal{C}}(\mathbf{v}) = [\cdots,I_{nm},\cdots,\cdots,J_{nm}, \cdots]^T.
\end{align}
The Jacobian $\mathbf{F}(\mathbf{v})$ can be derived from \eqref{power_flow_eq_PMU}, \eqref{power_flow_eq_SCADA_power_inj} and \eqref{power_flow_eq_SCADA_power_flow}
\begin{align}\label{Jacobian}
    \mathbf{F}(\mathbf{v})
    &=
    \begin{bmatrix}
    		\mathbf{I}_{2N}\\
		\mathbf{H}_{\mathcal{C}}\\
		(\mathbf{I}_{2N}\otimes\mathbf{v})^T\mathbf{H}_{\mathcal{I}}\\
		(\mathbf{I}_{4E}\otimes\mathbf{v})^T\mathbf{H}_{\mathcal{F}}
    \end{bmatrix}^T
\end{align}
where
\begin{align*}
	\mathbf{H}_{\mathcal{I}} &\triangleq [ \cdots,\mathbf{N}_{P,n}+\mathbf{N}_{P,n}^T,\cdots,\mathbf{N}_{Q,n}+\mathbf{N}_{Q,n}^T,\cdots]^T\\
	\mathbf{H}_{\mathcal{F}} &\triangleq[\cdots,\mathbf{E}_{P,nm}+\mathbf{E}_{P,nm}^T,\cdots,\mathbf{E}_{Q,nm}+\mathbf{E}_{Q,nm}^T,\cdots]^T\\	
	\mathbf{H}_{\mathcal{C}} &\triangleq [\cdots,\mathbf{H}_{I,n}^T,\cdots,\cdots,\mathbf{H}_{J,n}^T,\cdots]^T
\end{align*}
using $\mathbf{S}_n \triangleq \mathbf{I}_{E_n}\otimes\left(\mathbf{1}_2\otimes\mathbf{e}_n\right)^T$ with $E_n$ being the number of incident lines at bus $n$ and
\begin{align}
	\mathbf{H}_{I,n} &\triangleq \mathbf{S}_n \mathbf{C}_{I,n}, \quad  \mathbf{C}_{I,n} \triangleq [\cdots,\mathbf{C}_{I,nm}^T,\cdots]^T\\
	\mathbf{H}_{J,n} &\triangleq \mathbf{S}_n \mathbf{C}_{J,n},\quad \mathbf{C}_{J,n} \triangleq [\cdots,\mathbf{C}_{J,nm}^T,\cdots]^T.\nonumber
\end{align}

\section{Proof of Lemma 1}\label{proof_lem_Jacobian_lipschitz}

The $F$-norm inequality $\|\cdot\|\leq \|\cdot\|_F$ gives us
\begin{align}
	 \|\mathbf{\tilde{F}}_i(\mathbf{v})-\mathbf{\tilde{F}}_i(\mathbf{v}')\|^2\leq\|\mathbf{\tilde{F}}_i(\mathbf{v})-\mathbf{\tilde{F}}_i(\mathbf{v}')\|_F^2
\end{align}	
for all $i$. Since $\mathbf{\tilde{F}}_i(\mathbf{v})=\bdsb{\Gamma}_i^{-{1}/{2}}\mathbf{T}_i\mathbf{F}(\mathbf{v})$, we further use the multiplicative norm inequality $\left\|\mathbf{A}\mathbf{B}\right\|_F\leq \left\|\mathbf{A}\right\|_F\left\|\mathbf{B}\right\|_F$
\begin{align}
	\left\|\mathbf{\tilde{F}}_i(\mathbf{v})-\mathbf{\tilde{F}}_i(\mathbf{v}')\right\|_F^2
	&= \left\|\bdsb{\Gamma}_i^{-1/2}\mathbf{T}_i\left[\mathbf{F}(\mathbf{v})-\mathbf{F}(\mathbf{v}')\right]\right\|_F^2\\
	&\leq \lambda_{\min}^{-1}(\bdsb{\Gamma}_i)\left\|\mathbf{F}(\mathbf{v})-\mathbf{F}(\mathbf{v}')\right\|_F^2.
\end{align}	
From \eqref{Jacobian}, we have
\begin{align}
	\mathbf{F}(\mathbf{v})-\mathbf{F}(\mathbf{v}')
	=
	\begin{bmatrix}
		\mathbf{0}_{2N\times 2N}\\
		\mathbf{0}_{4E\times 2N}\\
		 \left[\mathbf{I}_{2N}\otimes(\mathbf{v}-\mathbf{v}')^T\right]\mathbf{H}_{\mathcal{I}}\\
		\left[\mathbf{I}_{4E}\otimes(\mathbf{v}-\mathbf{v}')^T\right]\mathbf{H}_{\mathcal{F}}
	\end{bmatrix}.
\end{align}
According to the $F$-norm definition $\left\|\mathbf{A}\right\|_F^2 = \mathrm{Tr}\left(\mathbf{A}^T\mathbf{A}\right)$ and the properties of the trace operator, we have for any $\mathbf{v}, \mathbf{v}'$
\begin{align*}
     \left\|\mathbf{F}(\mathbf{v})-\mathbf{F}(\mathbf{v}')\right\|_F^2
     &=
     \mathrm{Tr}\left[\left(\mathbf{I}_{2N}\otimes(\mathbf{v}-\mathbf{v}')(\mathbf{v}-\mathbf{v}')^T\right)\mathbf{H}_{\mathcal{I}}\mathbf{H}_{\mathcal{I}}^T\right]
     \\
     &~~~+
     \mathrm{Tr}\left[\left(\mathbf{I}_{4E}\otimes(\mathbf{v}-\mathbf{v}')(\mathbf{v}-\mathbf{v}')^T\right)\mathbf{H}_{\mathcal{F}}\mathbf{H}_{\mathcal{F}}^T\right].
\end{align*}
Expanding $\mathbf{H}_{\mathcal{I}}$ and $\mathbf{H}_{\mathcal{F}}$ in \eqref{Jacobian} and using their symmetric properties, we have
\begin{align}
	\left\|\mathbf{F}(\mathbf{v})-\mathbf{F}(\mathbf{v}')\right\|_F^2= (\mathbf{v}-\mathbf{v}')^T\mathbf{M}(\mathbf{v}-\mathbf{v}'),
\end{align}
where $\mathbf{M}=\mathbf{H}_{\mathcal{I}}^T\mathbf{H}_{\mathcal{I}}+\mathbf{H}_{\mathcal{F}}^T\mathbf{H}_{\mathcal{F}}$. It is well-known that any quadratic form of a symmetric matrix can be bounded as
\begin{align}\label{quadratic_form}
	 (\mathbf{v}-\mathbf{v}')^T\mathbf{M}(\mathbf{v}-\mathbf{v}')	\leq  \|\mathbf{M}\|\left\|\mathbf{v}-\mathbf{v}'\right\|^2.
\end{align}
The result follows by setting $\omega=\max_i~\sqrt{\|\mathbf{M}\|\lambda_{\min}^{-1}(\bdsb{\Gamma}_i)}$.

\bibliographystyle{IEEEtran}

\bibliography{E:/Dropbox/Shared/Simon_Anna/ref_general,E:/Dropbox/Shared/Simon_Anna/ref_power_system_SE,E:/Dropbox/Shared/Simon_Anna/ref_dist_opt,E:/Dropbox/Shared/Simon_Anna/ref_opt_PMU_placement,E:/Dropbox/Shared/Simon_Anna/ref_data_management}

\begin{thebibliography}{10}
\providecommand{\url}[1]{#1}
\csname url@samestyle\endcsname
\providecommand{\newblock}{\relax}
\providecommand{\bibinfo}[2]{#2}
\providecommand{\BIBentrySTDinterwordspacing}{\spaceskip=0pt\relax}
\providecommand{\BIBentryALTinterwordstretchfactor}{4}
\providecommand{\BIBentryALTinterwordspacing}{\spaceskip=\fontdimen2\font plus
\BIBentryALTinterwordstretchfactor\fontdimen3\font minus
  \fontdimen4\font\relax}
\providecommand{\BIBforeignlanguage}[2]{{%
\expandafter\ifx\csname l@#1\endcsname\relax
\typeout{** WARNING: IEEEtran.bst: No hyphenation pattern has been}%
\typeout{** loaded for the language `#1'. Using the pattern for}%
\typeout{** the default language instead.}%
\else
\language=\csname l@#1\endcsname
\fi
#2}}
\providecommand{\BIBdecl}{\relax}
\BIBdecl

\bibitem{li2012decentralized}
X.~Li, Z.~Wang, and A.~Scaglione, ``{D}ecentralized {D}ata {P}rocessing and
  {M}anagement in {S}mart {G}rid via {G}ossiping,'' in \emph{Sensor Array and
  Multichannel Signal Processing Workshop (SAM), 2012 IEEE 7th}.\hskip 1em plus
  0.5em minus 0.4em\relax IEEE, 2012, pp. 1--4.

\bibitem{li2013robust}
X.~Li and A.~Scaglione, ``{Robust Collaborative State Estimation for Smart Grid
  Monitoring},'' \emph{Acoustics, Speech and Signal Processing, 2013. ICASSP
  2013 Proceedings. 2013 IEEE International Conference on}.

\bibitem{schweppe1974static}
F.~Schweppe and E.~Handschin, ``{S}tatic {S}tate {E}stimation in {E}lectric
  {P}ower {S}ystems,'' \emph{Proceedings of the IEEE}, vol.~62, no.~7, pp.
  972--982, 1974.

\bibitem{yang2011transition}
T.~Yang, H.~Sun, and A.~Bose, ``{T}ransition to a {T}wo-{L}evel {L}inear
  {S}tate {E}stimator : {P}art i \& ii,'' \emph{IEEE Trans. Power Syst.},
  no.~99, pp. 1--1, 2011.

\bibitem{phadke1986state}
A.~Phadke, J.~Thorp, and K.~Karimi, ``{S}tate {E}stimation with {P}hasor
  {M}easurements,'' \emph{Power Engineering Review, IEEE}, no.~2, pp. 48--48,
  1986.

\bibitem{zivanovic1996implementation}
R.~Zivanovic and C.~Cairns, ``{I}mplementation of {PMU} {T}echnology in {S}tate
  {E}stimation: an {O}verview,'' in \emph{AFRICON, 1996., IEEE AFRICON 4th},
  vol.~2.\hskip 1em plus 0.5em minus 0.4em\relax IEEE, 1996, pp. 1006--1011.

\bibitem{meliopoulos2010supercalibrator}
A.~Meliopoulos, G.~Cokkinides, C.~Hedrington, and T.~Conrad, ``{T}he
  {S}upercalibrator-{A} {F}ully {D}istributed {S}tate {E}stimator,'' in
  \emph{Power and Energy Society General Meeting, 2010 IEEE}.\hskip 1em plus
  0.5em minus 0.4em\relax IEEE, 2010, pp. 1--8.

\bibitem{qin2007hybrid}
X.~Qin, T.~Bi, and Q.~Yang, ``{H}ybrid {N}on-linear {S}tate {E}stimation with
  {V}oltage {P}hasor {M}easurements,'' in \emph{Power Engineering Society
  General Meeting, 2007. IEEE}.\hskip 1em plus 0.5em minus 0.4em\relax IEEE,
  2007, pp. 1--6.

\bibitem{nuqui2007hybrid}
R.~Nuqui and A.~Phadke, ``{H}ybrid {L}inear {S}tate {E}stimation {U}tilizing
  {S}ynchronized {P}hasor {M}easurements,'' in \emph{Power Tech, 2007 IEEE
  Lausanne}.\hskip 1em plus 0.5em minus 0.4em\relax IEEE, 2007, pp. 1665--1669.

\bibitem{chakrabarti2010comparative}
S.~Chakrabarti, E.~Kyriakides, G.~Ledwich, and A.~Ghosh, ``{A} {C}omparative
  {S}tudy of the {M}ethods of {I}nclusion of {PMU} {C}urrent {P}hasor
  {M}easurements in a {H}ybrid {S}tate {E}stimator,'' in \emph{Power and Energy
  Society General Meeting, 2010 IEEE}.\hskip 1em plus 0.5em minus 0.4em\relax
  IEEE, 2010, pp. 1--7.

\bibitem{avila2009recent}
R.~Avila-Rosales, M.~Rice, J.~Giri, L.~Beard, and F.~Galvan, ``{R}ecent
  {E}xperience with a {H}ybrid {SCADA/PMU} {O}nline {S}tate {E}stimator,'' in
  \emph{Power \& Energy Society General Meeting, 2009}.\hskip 1em plus 0.5em
  minus 0.4em\relax IEEE, 2009, pp. 1--8.

\bibitem{brice1982multiprocessor}
C.~Brice and R.~Cavin, ``{M}ultiprocessor {S}tatic {S}tate {E}stimation,''
  \emph{IEEE Trans. Power App. Syst.}, no.~2, pp. 302--308, 1982.

\bibitem{kurzyn1983real}
M.~Kurzyn, ``{R}eal-{T}ime {S}tate {E}stimation for {L}arge-{S}cale {P}ower
  {S}ystems,'' \emph{IEEE Trans. Power App. Syst.}, no.~7, pp. 2055--2063,
  1983.

\bibitem{gomez2011multilevel}
A.~G{\'o}mez-Exp{\'o}sito, A.~Abur, A.~de~la Villa~Ja{\'e}n, and
  C.~G{\'o}mez-Quiles, ``{A} {M}ultilevel {S}tate {E}stimation {P}aradigm for
  {S}mart {G}rids,'' \emph{Proceedings of the IEEE}, no.~99, pp. 1--25, 2011.

\bibitem{falcao1995parallel}
D.~Falcao, F.~Wu, and L.~Murphy, ``{P}arallel and {D}istributed {S}tate
  {E}stimation,'' \emph{IEEE Trans. Power Syst.}, vol.~10, no.~2, pp. 724--730,
  1995.

\bibitem{lin1992distributed}
S.~Lin, ``{A} {D}istributed {S}tate {E}stimator for {E}lectric {P}ower
  {S}ystems,'' \emph{IEEE Trans. Power Syst.}, vol.~7, no.~2, pp. 551--557,
  1992.

\bibitem{ebrahimian2000state}
R.~Ebrahimian and R.~Baldick, ``{S}tate {E}stimation {D}istributed
  {P}rocessing,'' \emph{IEEE Trans. Power Syst.}, vol.~15, no.~4, pp.
  1240--1246, 2000.

\bibitem{van1981two}
T.~Van~Cutsem, J.~Horward, and M.~Ribbens-Pavella, ``{A} {T}wo-{L}evel {S}tatic
  {S}tate {E}stimator for {E}lectric {P}ower {S}ystems,'' \emph{IEEE Trans.
  Power App. Syst.}, no.~8, pp. 3722--3732, 1981.

\bibitem{zhao2005multi}
L.~Zhao and A.~Abur, ``{M}ulti-area {S}tate {E}stimation using {S}ynchronized
  {P}hasor {M}easurements,'' \emph{IEEE Trans. Power Syst.}, vol.~20, no.~2,
  pp. 611--617, 2005.

\bibitem{jiang2007distributed}
W.~Jiang, V.~Vittal, and G.~Heydt, ``{A} {D}istributed {S}tate {E}stimator
  {U}tilizing {S}ynchronized {P}hasor {M}easurements,'' \emph{IEEE Trans. Power
  Syst.}, vol.~22, no.~2, pp. 563--571, 2007.

\bibitem{xie2012fully}
L.~Xie, D.~Choi, S.~Kar, and H.~Poor, ``{F}ully {D}istributed {S}tate
  {E}stimation for {W}ide-{A}rea {M}onitoring {S}ystems,'' \emph{Smart Grid,
  IEEE Transactions on}, vol.~3, no.~3, pp. 1154--1169, 2012.

\bibitem{kekatos2012distributed}
V.~Kekatos and G.~Giannakis, ``{D}istributed {R}obust {P}ower {S}ystem {S}tate
  {E}stimation,'' \emph{Arxiv preprint arXiv:1204.0991}, 2012.

\bibitem{lopes2008diffusion}
C.~Lopes and A.~Sayed, ``{D}iffusion {L}east-{M}ean {S}quares over {A}daptive
  {N}etworks: {F}ormulation and {P}erformance {A}nalysis,'' \emph{IEEE Trans.
  Signal Process.}, vol.~56, no.~7, pp. 3122--3136, 2008.

\bibitem{nedic2009distributed}
A.~Nedic and A.~Ozdaglar, ``{D}istributed {S}ubgradient {M}ethods for
  {M}ulti-agent {O}ptimization,'' \emph{Automatic Control, IEEE Transactions
  on}, vol.~54, no.~1, pp. 48--61, 2009.

\bibitem{kar2008distributed}
S.~Kar, J.~Moura, and K.~Ramanan, ``{Distributed Parameter Estimation in Sensor
  Networks: Nonlinear Observation Models and Imperfect Communication},''
  \emph{Information Theory, IEEE Transactions on}, vol.~58, no.~6, pp. 3575
  --3605, june 2012.

\bibitem{van1985bad}
T.~Van~Cutsem, M.~Ribbens-Pavella, and L.~Mili, ``{B}ad {D}ata {I}dentification
  {M}ethods in {P}ower {S}ystem {S}tate {E}stimation-{A} {C}omparative
  {S}tudy,'' \emph{Power Apparatus and Systems, IEEE Transactions on}, no.~11,
  pp. 3037--3049, 1985.

\bibitem{garcia1979fast}
A.~Garcia, A.~Monticelli, and P.~Abreu, ``{F}ast {D}ecoupled {S}tate
  {E}stimation and {B}ad {D}ata {P}rocessing,'' \emph{Power Apparatus and
  Systems, IEEE Transactions on}, no.~5, pp. 1645--1652, 1979.

\bibitem{bobba2010detecting}
R.~Bobba, K.~Rogers, Q.~Wang, H.~Khurana, K.~Nahrstedt, and T.~Overbye,
  ``{D}etecting {F}alse {D}ata {I}njection {A}ttacks on {DC} {S}tate
  {E}stimation,'' in \emph{Preprints of the First Workshop on Secure Control
  Systems, CPSWEEK 2010}, 2010.

\bibitem{liu2011false}
Y.~Liu, P.~Ning, and M.~Reiter, ``{F}alse {D}ata {I}njection {A}ttacks against
  {S}tate {E}stimation in {E}lectric {P}ower {G}rids,'' \emph{ACM Transactions
  on Information and System Security (TISSEC)}, vol.~14, no.~1, p.~13, 2011.

\bibitem{chen2006placement}
J.~Chen and A.~Abur, ``{P}lacement of {PMU}s to {E}nable {B}ad {D}ata
  {D}etection in {S}tate {E}stimation,'' \emph{Power Systems, IEEE Transactions
  on}, vol.~21, no.~4, pp. 1608--1615, 2006.

\bibitem{giani2011smart}
A.~Giani, E.~Bitar, M.~Garcia, M.~McQueen, P.~Khargonekar, and K.~Poolla,
  ``{S}mart {G}rid {D}ata {I}ntegrity {A}ttacks: {C}haracterizations and
  {C}ountermeasures $\pi$,'' in \emph{Smart Grid Communications
  (SmartGridComm), 2011 IEEE International Conference on}.\hskip 1em plus 0.5em
  minus 0.4em\relax IEEE, 2011, pp. 232--237.

\bibitem{monticelli1999state}
A.~Monticelli, ``{S}tate {E}stimation in {E}lectric {P}ower {S}ystems: A
  {G}eneralized {A}pproach, 1999.''

\bibitem{choi2011fully}
D.~Choi and L.~Xie, ``Fully distributed bad data processing for wide area state
  estimation,'' in \emph{Smart Grid Communications (SmartGridComm), 2011 IEEE
  International Conference on}.\hskip 1em plus 0.5em minus 0.4em\relax IEEE,
  2011, pp. 546--551.

\bibitem{bin1994implementable}
S.~Bin and C.~Lin, ``{A}n {I}mplementable {D}istributed {S}tate {E}stimator and
  {D}istributed {B}ad {D}ata {P}rocessing {S}chemes for {E}lectric {P}ower
  {S}ystems,'' \emph{Power Systems, IEEE Transactions on}, vol.~9, no.~3, pp.
  1277--1284, 1994.

\bibitem{pasqualetti2011distributed}
F.~Pasqualetti, R.~Carli, and F.~Bullo, ``{A} {D}istributed {D}ethod for
  {S}tate {E}stimation and {F}alse {D}ata {D}etection in {P}ower {N}etworks,''
  in \emph{Smart Grid Communications (SmartGridComm), 2011 IEEE International
  Conference on}.\hskip 1em plus 0.5em minus 0.4em\relax IEEE, 2011, pp.
  469--474.

\bibitem{xu2011sparse}
W.~Xu, M.~Wang, and A.~Tang, ``{S}parse {R}ecovery from {N}onlinear
  {M}easurements with {A}pplications in {B}ad {D}ata {D}etection for {P}ower
  {N}etworks,'' \emph{arXiv preprint arXiv:1112.6234}, 2011.

\bibitem{merrill1971bad}
H.~Merrill and F.~Schweppe, ``{B}ad {D}ata {S}uppression in {P}ower {S}ystem
  {S}tatic {S}tate {E}stimation,'' \emph{Power Apparatus and Systems, IEEE
  Transactions on}, no.~6, pp. 2718--2725, 1971.

\bibitem{falcao1982power}
D.~Falc{\~a}o, P.~Cooke, and A.~Brameller, ``{P}ower {S}ystem {T}racking
  {S}tate {E}stimation and {B}ad {D}ata {P}rocessing,'' \emph{Power Apparatus
  and Systems, IEEE Transactions on}, no.~2, pp. 325--333, 1982.

\bibitem{lo1983development}
K.~Lo, P.~Ong, R.~McColl, A.~Moffatt, and J.~Sulley, ``{D}evelopment of a
  {S}tatic {S}tate {E}stimator {P}art i: {E}stimation and {B}ad {D}ata
  {S}uppression,'' \emph{Power Apparatus and Systems, IEEE Transactions on},
  no.~8, pp. 2486--2491, 1983.

\bibitem{handschin1975bad}
E.~Handschin, F.~Schweppe, J.~Kohlas, and A.~Fiechter, ``{B}ad {D}ata
  {A}nalysis for {P}ower {S}ystem {S}tate {E}stimation,'' \emph{Power Apparatus
  and Systems, IEEE Transactions on}, vol.~94, no.~2, pp. 329--337, 1975.

\bibitem{li2012convergence}
X.~Li and A.~Scaglione, ``{Convergence and Applications of a Gossip-based
  Gauss-Newton Algorithm},'' \emph{submitted to IEEE Trans. on Signal
  Processing, arXiv preprint arXiv:1210.0056}, 2012.

\bibitem{lavaei2010zero}
J.~Lavaei and S.~Low, ``{Z}ero {D}uality {G}ap in {O}ptimal {P}ower {F}low
  {P}roblem,'' \emph{IEEE Transactions on Power Systems}, 2010.

\bibitem{kayfundamentals}
S.~Kay, ``{F}undamentals of {S}tatistical {S}ignal {P}rocessing : {V}olume {I}
  \& {II}, 1993.''

\bibitem{tsitsiklis1984problems}
J.~Tsitsiklis, ``{P}roblems in {D}ecentralized {D}ecision {M}aking and
  {C}omputation.'' DTIC Document, Tech. Rep., 1984.

\bibitem{blondel2005convergence}
V.~Blondel, J.~Hendrickx, A.~Olshevsky, and J.~Tsitsiklis, ``{C}onvergence in
  {M}ultiagent {C}oordination, {C}onsensus, and {F}locking,'' in \emph{Decision
  and Control, 2005 and 2005 European Control Conference. CDC-ECC'05. 44th IEEE
  Conference on}.\hskip 1em plus 0.5em minus 0.4em\relax IEEE, 2005, pp.
  2996--3000.

\bibitem{boyd2006randomized}
S.~Boyd, A.~Ghosh, B.~Prabhakar, and D.~Shah, ``{R}andomized {G}ossip
  {A}lgorithms,'' \emph{IEEE Trans. Inf. Theory}, vol.~52, no.~6, pp.
  2508--2530, 2006.

\bibitem{clements1990observability}
K.~Clements, ``{O}bservability {M}ethods and {O}ptimal {M}eter {P}lacement,''
  \emph{International Journal of Electrical Power \& Energy Systems}, vol.~12,
  no.~2, pp. 88--93, 1990.

\bibitem{eriksson2004applied}
K.~Eriksson, D.~Estep, and C.~Johnson, \emph{{A}pplied {M}athematics, {B}ody
  and {S}oul: {D}erivates and {G}eometry in $\mathbb{R}^3$}.\hskip 1em plus
  0.5em minus 0.4em\relax Springer Verlag, 2004, vol.~3.

\bibitem{kekatos2011convex}
V.~Kekatos and G.~Giannakis, ``{A} {C}onvex {R}elaxation {A}pproach to
  {O}ptimal {P}lacement of {P}hasor {M}easurement {U}nits,'' in
  \emph{Computational Advances in Multi-Sensor Adaptive Processing (CAMSAP),
  2011 4th IEEE International Workshop on}.\hskip 1em plus 0.5em minus
  0.4em\relax IEEE, 2011, pp. 145--148.

\bibitem{li2011phasor}
Q.~Li, R.~Negi, and M.~Ilic, ``{Phasor Measurement Units Placement for Power
  System State Estimation: a Greedy Approach},'' in \emph{Power and Energy
  Society General Meeting, 2011 IEEE}.\hskip 1em plus 0.5em minus 0.4em\relax
  IEEE, 2011, pp. 1--8.

\bibitem{li2012information}
Q.~Li, T.~Cui, Y.~Weng, R.~Negi, F.~Franchetti, and M.~Ilic, ``{An
  Information-Theoretic Approach to PMU Placement in Electric Power Systems},''
  \emph{Arxiv preprint arXiv:1201.2934}, 2012.

\bibitem{UK_grid}
\BIBentryALTinterwordspacing
``{U}. {K}. {N}ational {G}rid-{R}eal {T}ime {O}perational {D}ata,'' 2009,
  [Online; accessed 22-July-2004]. [Online]. Available:
  \url{http://www.nationalgrid.com/uk/Electricity/Data/}
\BIBentrySTDinterwordspacing

\end{thebibliography}


\begin{IEEEbiography}[{\includegraphics[width=1in,height=1.25in,clip,keepaspectratio]{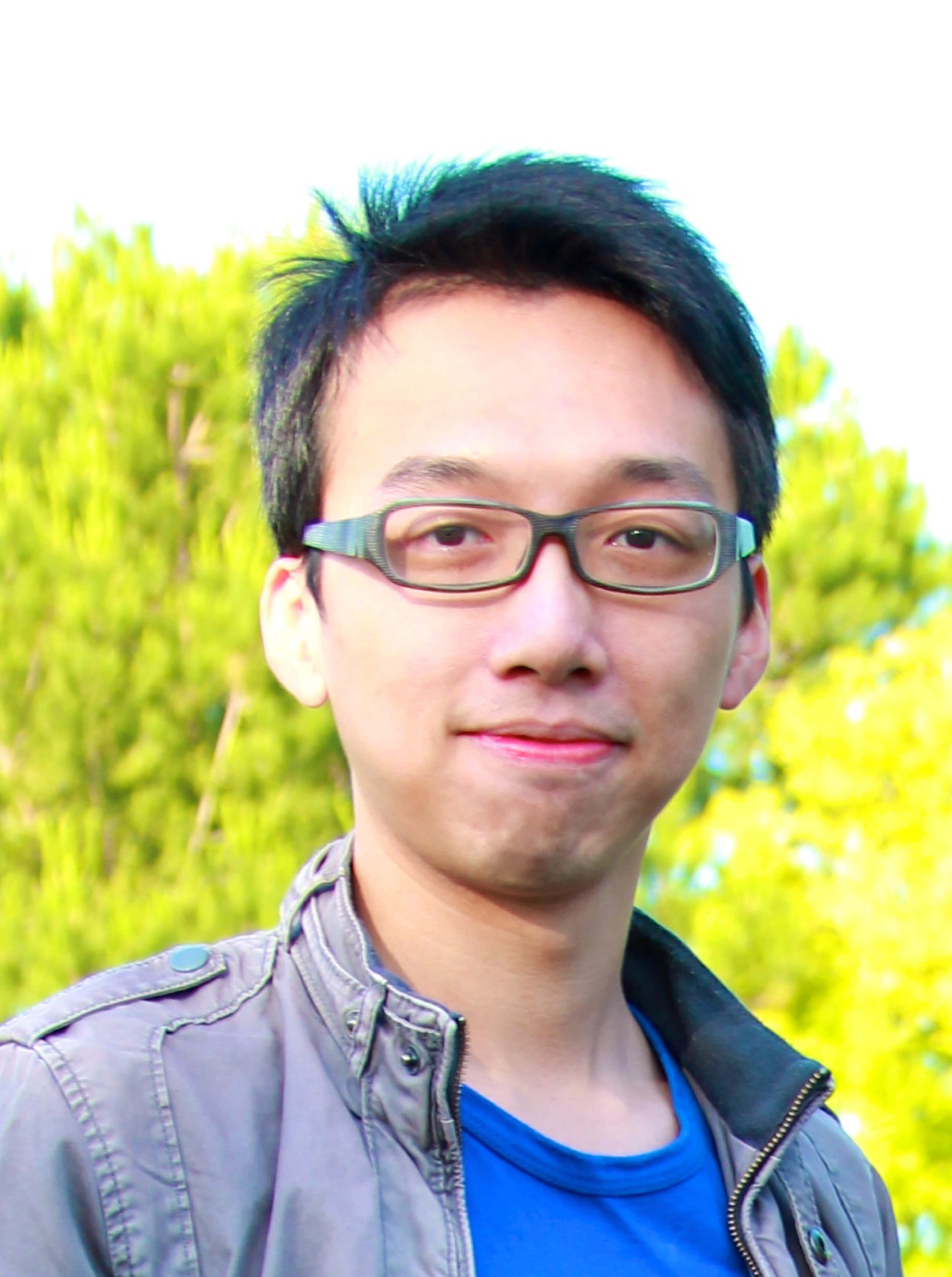}}]
{Xiao Li} has been pursuing the Ph.D. degree at the University of California, Davis, since 2009. He received the B.Eng. degree (2007) from Sun Yat-Sen (Zhongshan) University, China, and the M.Phil. degree (2009) from the University of Hong Kong. His research interests are in the theoretical and algorithmic studies in signal processing and optimizations, statistical learning and inferences for high dimensional data, distributed optimizations and adaptive algorithms, as well as their applications in communications, networked systems, and smart grid.
\end{IEEEbiography}

\vspace{-0.3cm}
\begin{IEEEbiography}[{\includegraphics[width=1in,height=1.25in,clip,keepaspectratio]{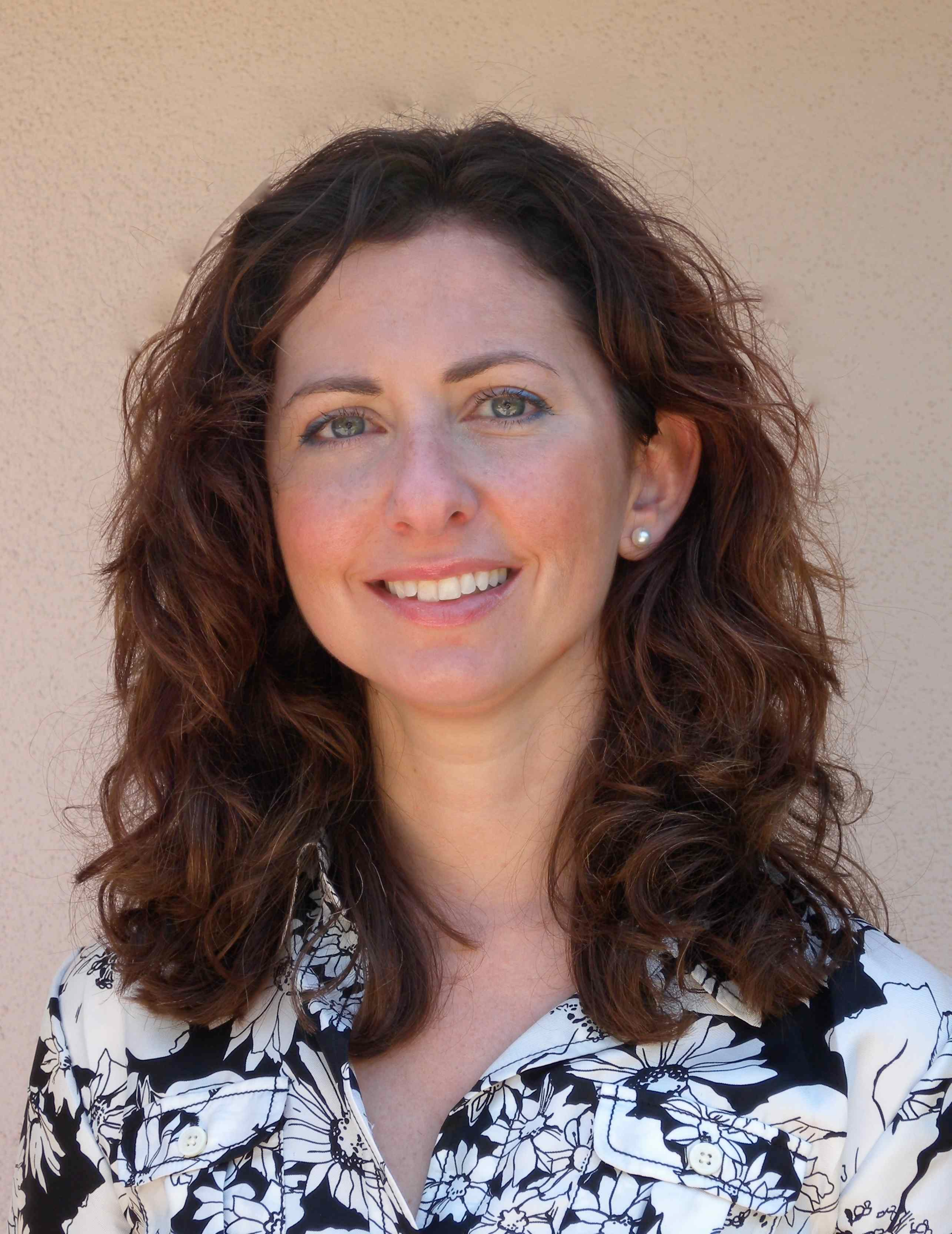}}]
{\bf Anna Scaglione} (IEEE F10') received the Laurea (M.Sc. degree) in 1995 and the Ph.D. degree in 1999 from the University of Rome, ``La Sapienza." She is currently Professor in Electrical and Computer Engineering at University of California at Davis, where she joined in 2008. She was previously at Cornell University, Ithaca, NY, from 2001 where became Associate Professor in 2006; prior to joining Cornell she was Assistant Professor in the year 2000-2001, at the University of New Mexico.

She served as Associate Editor for the \textsc{IEEE Transactions on Wireless Communications} from 2002 to 2005, and serves since 2008 the Editorial Board of the \textsc{IEEE Transactions on Signal Processing} from 2008 to 2011, where she is Area Editor. Dr. Scaglione received the {2000 IEEE Signal Processing Transactions Best Paper Award} the {NSF Career Award} in 2002 and she is co-recipient of the {Ellersick Best Paper Award} (MILCOM 2005) and the {2013 IEEE Donald G. Fink Prize Paper Award}. Her expertise is in the broad area of signal processing for communication systems and networks. Her current research focuses on communication and wireless networks, sensors' systems for monitoring, control and energy management  and network science.
\end{IEEEbiography}

\end{document}